\documentclass[onecolumn]{IEEEtran}
\usepackage{cite}
\usepackage{amsmath,amsthm}
\usepackage{amssymb,amsfonts}
\usepackage{algorithmic}
\usepackage{graphicx}
\usepackage{textcomp}
\usepackage{marvosym}
\usepackage{mathtools}
\usepackage{wrapfig}
\usepackage{mathrsfs}
\usepackage[norelsize]{algorithm2e}
\newtheorem{theorem}{Theorem}[section]
\newtheorem{lem}{Lemma}[section]

\newtheorem{definition}{Definition}
\newtheorem{assump}{Assumption}
\newtheorem{prop}{Proposition}[section]
\newtheorem{corollary}{Corollary}[section]
\newtheorem{remark}{Remark}

\newcommand{\G}{\mathcal{G}}
\newcommand{\V}{\mathcal{V}}
\newcommand{\E}{\mathcal{E}}

\newcommand{\argmin}{\operatornamewithlimits{argmin}}

\title{Convex Decreasing Algorithms: Distributed Synthesis and Finite-time Termination in Higher Dimension}
\author{James Melbourne$^*$\thanks{$^*$James Melbourne, Vivek Khatana, Sourav Patel, and Murti V. Salapaka are with the Department of Electrical and Computer Engineering, University of Minnesota, Minneapolis, USA.}, Govind Saraswat$^\dagger$\thanks{$^\dagger$Govind Saraswat is with 
the National Renewable Energy Laboratory, Golden, CO, USA (e-mail: govind.saraswat@nrel.gov).}, Vivek Khatana,\\ Sourav Patel, and Murti V. Salapaka$^\ddagger$ 
\thanks{$^\ddagger$This work was authored in part by the National Renewable Energy Laboratory, operated by Alliance for Sustainable Energy, LLC, for the U.S. Department of Energy (DOE) under Contract No. DE-AC36-08GO28308. Funding provided by the Advanced Research Projects Agency-Energy (ARPA-E) under grant no. DE-AR0000701. The views expressed in the article do not necessarily represent the views of the DOE or the U.S. Government. The U.S. Government retains and the publisher, by accepting the article for publication, acknowledges that the U.S. Government retains a nonexclusive, paid-up, irrevocable, worldwide license to publish or reproduce the published form of this work, or allow others to do so, for U.S. Government purposes.  }
}
\begin{document}
\maketitle

\begin{abstract}
We introduce a general mathematical framework for distributed algorithms, and a monotonicity property frequently satisfied in application.  These properties are leveraged to provide finite-time guarantees for converging algorithms, suited for use in the absence of a central authority.  A central application is to consensus algorithms in higher dimension.  These pursuits motivate a new peer to peer convex hull algorithm which we demonstrate to be an instantiation of the described theory.  To address the diversity of convex sets and the potential computation and communication costs of knowing such sets in high dimension, a lightweight norm based stopping criteria is developed.  More explicitly, we give a distributed algorithm that terminates in finite time when applied to consensus problems in higher dimensions and guarantees the convergence of the consensus algorithm in norm, within any given tolerance.  Applications to consensus least squared estimation and distributed function determination are developed.  The practical utility of the algorithm is illustrated through MATLAB simulations.
\end{abstract}


\section{Introduction}
Recent advancements in the field of intelligent communication technologies and ad-hoc wireless sensor networks for monitoring and control of multi-agent systems have necessitated design of algorithms to assert global information without the complete knowledge of the system.
The task of obtaining the global information is often accomplished via a class of algorithms that employ strategies of achieving consensus. In a consensus algorithm, agents iteratively and in a distributed manner agree on a common state. The ideas of distributed consensus algorithms can be traced back to the seminal works, see  \cite{arrow1958decentralization,degroot1974reaching,lynch1996distributed,tsitsiklis1984problems}. Recent works on consensus algorithms are focused on designing protocols to drive agents to the average of their initial states, see  \cite{kempe2003gossip,dominguez2010coordination,hadjicostis2013average,cai2012average}. These protocols were designed for cases where the state of each agent is a scalar value. However, increasing storage and computation capabilities of modern day sensor interfacing technologies have motivated large-scale applications, examples of which include distributed machine learning, see  \cite{predd2009collaborative}, multi-agent control and co-ordination, see  \cite{fax2002information,olfati2007consensus}, distributed optimization problems, see  \cite{nedic2014distributed,khatana2019gradient},  distributed sensor localization, see  \cite{khan2009distributed}. In order to meet the requirements of such applications there is need of distributed consensus algorithms that allow for vector states. \cite{khan2010higher} presents such a higher dimensional consensus protocol. The framework in \cite{khan2010higher} is based on a leader-follower architecture with the agents being partitioned between anchors and sensors. Anchors are agents with  fixed states behaving as leaders in the algorithm, while the sensors change their state by taking a convex combination of their state with the neighboring  nodes' states.

Cyber physical systems such as electrical power networks need to accommodate for large number of states for crucial applications such as state-estimations, optimal dispatch, and demand response for ancillary services. \cite{patel2017distributed} formulates the distributed apportioning problem using consensus protocols where only a single state is shared for each protocol.
A similar situation involving higher dimensional states arises in distributed resource allocation problems where a fixed amount of resource is required to be apportioned among all participating agents in a network  and a convex cost associated with the each resource is to be distributively minimized. The method used to solve the above resource allocation problem involves a higher dimensional consensus protocol (see  \cite{nedic2010constrained}).

Other applications where accommodating higher dimensional data is required are distributed optimization and their applications to deep neural networks such as in the case of diffusion learning where local agents perform model training using local datasets in parallel, for example see \cite{chakraborty2017deep}, where after the completion of the training process, model parameters are then shared with the neighboring nodes allowing for faster dissemination of the model parameters. This method also allows for asynchronous learning, where different agents learn with models of different parameters at every iteration. Recent frameworks for unsupervised learning such as Generative Adversarial Networks (GANs) where the objective is to train learning models to be robust against adversarial attacks (see  \cite{goodfellow2014generative}). \cite{kurakin2016adversarial} propose methods for scaling adversarial training and learning to large datasets performed on a distributed setup where each agent trains a discriminator model on a local training data and a shared generator model is trained based on the feedback received from the discriminator models of the different agents. These applications need sharing of a high dimensional parameter vector, usually of the order of $32 \times 32 \times 3 \sim 1024 \times 1024 \times 3$. Applications such as spectrum sensing in \textit{Ad hoc} cognitive radio networks, see  \cite{li2009distributed}, distributed detection of malicious attacks in finite-time consensus networks \cite{patel2020distributed} and control of autonomous agents like unmanned aerial vehicles (UAVs) for search and survey operations, see  \cite{fax2002information} also depend on implementing higher dimensional consensus protocols. 

Termination of consensus algorithms in finite time provides an advantage of getting an approximate consensus while saving valuable computation and communication resources. For the scalar average consensus protocols discussed earlier, the authors in \cite{yadav2007distributed,ratio_consensus_lab} have proposed such a finite time stopping criteria utilizing two additional states namely the global maximum and global minimum over the network. This allows each agent to distributively detect the convergence to the (approximate) average and terminate further computations. The works in  \cite{saraswat2019distributed,prakash2019distributed} generalize this result to the cases of dynamic interconnection topology and communication delays. The authors in \cite{sundaram2007finite} present a method based on the minimal polynomial associated with the weight matrix in the state update iterations to achieve the consensus value in a finite number of iterations. However, to calculate the coefficients of the minimal polynomial each node has to run $n$ (total number of agents) different linear iterations each for at least $n+1$ time-steps.

The protocol developed here is based on a new geometric insight into the behavior of push-sum or ratio consensus algorithms that is of independent interest.  We prove that many popular consensus algorithms fit into a general class of algorithms which we call ``convex decreasing''.  We demonstrate that convex decreasing consensus algorithms satisfy monotone convergence properties, which we leverage to give guarantees on network configurations.  That is, if it can be determined that all nodes of convex decreasing consensus algorithm belongs to a convex set, then all nodes of the algorithm will remain in said convex set.  We develop distributed stopping criteria based on this geometric insight.

Centralized algorithms for finding the convex hull of a finite set of $n$ points in a plane have been long proposed. Such algorithms have the worst-case running time of $\mathcal{O}(n~\log n)$, which is also the best achievable performance for obtaining the ordered hull (see \cite{preparata1979optimal}). However, with the recent advancements in distributed multi-agent systems, the problem of estimating the convex hull in a distributed manner, in the absence of a centralized entity, has become important.
To this regard,  distributed algorithms to estimate convex hull have been proposed in the literature for applications such as in classification problems, locational region estimation and formation control to name a few. In classification problems, \cite{fernandez2017one,casale2014approximate} have proposed an one-class scaled distributed convex hull algorithm where the geometric structure of the convex hull is used to model the boundary of the target class defining the one-class. However an approximation of the convex hull in the original large dimensional space by means of  randomly projected decisions on 2-D spaces results in a residual error when applied to finite time applications as analyzed in  \cite{kavan2006fast}. 
In order to obtain the convex hull of the feature space (kernel space) by communicating the extremities, \cite{osuna2002convex} employed a quadratic programming approach. However, the computational complexity of the proposed solution is limiting when extended to higher dimensions. \cite{kim2015distributed} proposed a convex hull algorithm, but require the assumption that the feature space is generated from a Gaussian mixture model and thus is limited to applications of a special class of support vector machines.   A difference in paradigm between the convex hull estimation pursued here and the literature on distributed programming for computation of convex hulls, is that the current article constructs a protocol that can be implemented in a plug and play manner in the absence of a central authority or knowledge of the network.  

The problem of computing a specified function of the sensor measurements is common in wireless sensor networks, see \cite{giridhar2005computing}, \cite{sundaram2008distributed}. The setup in \cite{giridhar2005computing} focuses on the problem of determining an arbitrary function of the sensor measurements from a specified sink node. The article studies the maximum rate at which the function can be calculated and communicated to the sink node. The article provides a characterization of the achievable rates for thee different classes of functions. The authors in \cite{sundaram2008distributed} proposed a method to calculate any arbitrary function in networks utilizing a linear iteration. The authors show that the proposed method can be modeled as a dynamical system. Based on the structured system observability approach it is shown that the linear iterations can determine any specified function of the initial values of the nodes after observing the evolution of the agent values for a large but finite time-steps. However, the method is based on forming observability  matrices for all agents in the network; here the method poses limitations as the size of the network increases leading to a large amount of computation and storage requirement. Further, the linear iterations require doubly-stochastic matrices for the agent state updates. The need of doubly-stochastic matrices makes the method in \cite{sundaram2008distributed} not applicable to directed networks. In this article we propose a distributed algorithm to determine any arbitrary function of the initial values of the agents, applicable to general connected directed graph topologies. Moreover, the proposed method has a fixed storage requirement and the communication overhead of the proposed method compared to the existing methods is insignificant.

In this article, we present a distributed stopping criteria for the higher dimensional consensus problem. Our first progress in this direction can be found in \cite{melbourne2020geometry}, where we investigated distributed stopping algorithm for the special case of ``ratio consensus''. Here we present a general theory of convex monotonicity, and demonstrate that one can apply the  distributed termination techniques to any algorithm satisfying this criteria.   In particular we show that in popular consensus algorithms (ratio consensus and row stochastic for example), the evolution of the convex hull of network states (in any dimension) indexed by time form a nested sequence of convex sets. This motivates an algorithm for distributively computing the convex hull within a time that scales linear with diameter of the network. We further provide a simpler algorithm which guarantees the convergence of consensus algorithm in norm, within any given tolerance.

\textit{ Statement of contribution:}
\begin{enumerate}
    \item This article constructs a general mathematical framework for convergence of network algorithms. In particular a notion of monotoniticity satisfied by important consensus algorithms is introduced.  We show ratio consensus \cite{kempe2003gossip} and row stochastic updating of a network to be examples of ``convex decreasing'' algorithms.
    
    \item  A convex hull algorithm, of independent interest, is developed for distributed determination of the extreme points of a set of vectors in the absence of a central authority.  In the context of a distributed convex decreasing algorithm, the hull algorithm can be used by the agents to obtain the convex hull by a fixed time $t$, and thus give guarantees on the state of the convex decreasing algorithm for all times $t' \geq t$.  
    
    \item Feasibility concerns for convex hull computation are addressed for high dimensional data, and an alternative lightweight (in the sense of both computational and communication cost) stopping criterion is given that guarantees finite convergence within an $\varepsilon$-threshold of consensus with respect to an arbitrary norm.
    
    \item As application of the theory developed, new stopping criteria are developed for consensus based least square estimation as well as distributed function calculation that give the agents convergence guarantees in a peer to peer network. 

\end{enumerate}

The rest of the paper is organized as follows. In Section~\ref{sec:DefProb}, the basic definitions needed for subsequent developments are presented. Further, we discuss the setup for the distributed average consensus in higher dimensions (called the vector consensus problem) using ratio consensus. Sections~\ref{sec:ConHull} presents an analysis on the polytopes of the network states generated in the ratio consensus algorithm.
Section~\ref{sec:normBase} establishes a norm-based finite-time termination criterion for the vector consensus problem. Theoretical findings are validated with simulations presented in Section~\ref{sec:results} followed by conclusions in Section~\ref{sec:Conclusion}.

\section{Definitions, and Problem Statement}\label{sec:DefProb}

\begin{subsection}{Definitions and Notations}
In this section we present basic notions of graph theory and linear algebra which are essential for the subsequent developments. Detailed description of graph theory and linear algebra notions are available in \cite{Die06}, and \cite{horn2012matrix} respectively.  We will also develop a general mathematical framework for consensus algorithms, and introduce a notion of convex monotonicity which will be crucial in the development of stopping criterion for vector consensus.

\begin{definition}(Cardinality of a set)
Let $A$ be a set. The cardinality of a set $A$ denoted by $|A|$ is the number of elements of the set $A$. 
\end{definition}

\begin{definition}(Directed Graph)
A directed graph (denoted as digraph) $\G$ is a pair $(\V,\E)$ where $\V$ is a set of vertices or nodes and $\E$ is a set of edges, which are ordered subsets of two distinct elements of $\V$. If an edge from $j \in \V$ to $i \in \V$ exists then it is denoted as $(i,j)\in \E$.  
\end{definition}
\begin{definition}(Path) 
In a directed graph, a directed path from node $i$ to $j$
exists if there is a sequence of distinct directed edges of $\G$ of
the form $(k_{1},i),(k_{2},k_{1}),...,(j,k_{m}).$ For the rest of the article, a path refers to a directed path.
\end{definition}

\begin{definition}(Path Length)
The path length, or length of a path is the number of directed edges belonging to the path.  By convention, we consider a node $i \in V$ to be connected to itself by a path of length zero.

\end{definition}
\begin{definition}(Strongly Connected Graph) A directed graph is strongly connected if to every $i,j \in V$ there exists a directed
path from node $i$ to node $j$.
\end{definition}
\begin{definition}(In-Neighborhood) \label{def: in-neighborhood} Set of in-neighbors of node $i \in \mathcal{V}$ is denoted by $N^-_i = \{j \ | \ (i,j)\in \mathcal{E}\}.$ In this article, we assume $(i,i) \in N^-_i$ for all $i \in \V.$
\end{definition}

\begin{definition}($m$-In-Neighborhood)
For $m=0$ define $N^-_i(0) = i$, and for $m >0$ define $N^-_i(m)$ to be $\cup_{j \in N_i^-(m-1)} N^-_j$, so that $N^-_i(m)$ is the set of nodes $j \in G$ from which $i$ can be reached in $m$ or less steps.
\end{definition}

\begin{definition}(Diameter of a Graph)
The diameter of a graph is the longest shortest path between any two nodes in the network. We will consider $D$ as an upper bound on the diameter of the graph throughout the rest of the article.
\end{definition}

\begin{definition}(Network State)
    For a vector space $W$, a $W$-valued network state is a function $f: {\mathcal{V}} \to W$, which we denote by $W^{\mathcal{V}}$.
\end{definition}

\begin{definition}(Network Update)
      A network update on $G = G(\mathcal{V},\mathcal{E})$, is a map $\phi: W^{\mathcal{V}} \to W^{\mathcal{V}}$.  
\end{definition}

\begin{definition}(Network and Consensus algorithms)
    A discrete time network algorithm on $G = G(\mathcal{V},\mathcal{E})$, is a finite or countably infinite sequence of maps of network updates.   When $W$ is endowed with a norm $\| \cdot \|$, a consensus algorithm is a network algorithm such that for any $f \in W^{\mathcal{V}}$,  \begin{align}\label{eq: phi_n update}
    \Phi_n (f) \coloneqq \phi_n \circ \phi_{n-1} \circ \cdots \circ \phi_1(f)
    \end{align}
    satisfies,
    \begin{align*}
       \lim_{n \to \infty} \| \Phi_n(f)(i)-  \Phi_n(f)(j) \| = 0,
    \end{align*}
    for all $i,j \in V$.
\end{definition}

We will only consider discrete time network algorithms in this work, a discrete time consensus algorithm will be referred to a consensus algorithm hereafter. 

\begin{definition}(Distributed Network Update and Algorithm)
    A network update $\phi$, is distributed if $f,g \in W^{\mathcal{V}}$ satisfying $f(j) = g(j)$ for $j \in N^-_i(1)$ implies 
    $
        \phi(f)(i) = \phi(g)(i).
   $
    A network algorithm $\{\phi_k\}_k$, is distributed if $\phi_k$ is a distributed network update for every $k$.
\end{definition}

For $i \in \mathcal{V}$, define $\pi_i : W^{\mathcal{V}} \to W^{N_i^-(i)}$ to be the restriction of $f$ to $N_i^-(1)$.  Explicitly, for $f \in W^{\mathcal{V}}$ and $j \in N_i^-(i)$, $\pi_{i}(f)(j) \coloneqq f(j)$.  Also, define $\epsilon_i : W^{N_i^-(i)} \to W^{\mathcal{V}}$ by
\begin{align*}
    \epsilon_i(x)(j) = \begin{cases}
                x(j) &\text{ when } j \in N_i^-(1),
                    \\
                0 &\text{ else, }
                \end{cases}
\end{align*}
for $x \in W^{N_i^-(1)}$.

In the following proposition, we gives as alternative formulation of the fact that distributed updates are determined locally.

\begin{prop}
 A network update $\phi$ is distributed if and only if for every $i$, $f \mapsto \phi(f)(i)$ can be expressed as  $\psi_{i}(\pi_{i}(f)) $ for a function $\psi_{i}: W^{N_i^-(1)} \to W$. 
\end{prop} 
\begin{proof}
First suppose a function $\psi_{i}$ satisfying $\psi_{i}(\pi_i(f)) = \phi(f)(i)$ exists, and that $f$ and $g \in W^{\mathcal{V}}$ satisfy, $f_j = g_j$ for $j \in N_i^-(1)$, then $\pi_i(f) = \pi_i(g)$. Thus $\phi(f)(i) = \psi_{i}(\pi_i(f)) = \psi_{i}(\pi_i(g)) = \phi(g)(i)$.  Conversely, assume $\phi$ is a distributed network update and define for $x \in W^{N_i^-(1)}$ $\psi_i(x) = \phi(\epsilon_i(x))(i)$. Observe that for $j \in N_i^-(1)$, $f(j) = \epsilon_i({\pi_{i}(f)})(j)$, hence by the definition of a distributed network update \begin{align}\label{eq: use of distributed defn} \phi(f)(i) = \phi(\epsilon_i({\pi_{i}(f)}))(i).
\end{align}Further by the definition of the function $\psi_{i}$,
\begin{align} \label{eq: defn of psi function}
 \phi(\epsilon_i({\pi_{i}(f)}))(i) = \psi_{i}(\pi_{i}(f)).
\end{align}
Combining \eqref{eq: use of distributed defn} and \eqref{eq: defn of psi function} completes the proof.
\end{proof}

Observe that in the case that the $\psi_{i}$ is linear in $f$ in the sense that $\psi_{i}(f) = \sum_{j \in N_i^-(1)} \lambda_{ij} f_j$ for $\lambda_j \in \mathbb{R}$, then $\phi_k(f)$ can be represented by a matrix ${ \bf \Lambda_k} = \left( \lambda_{ij} \right)$.  Conversely, $\phi$ represented by matrices, clearly induce linear $\psi_{i}$.

We will be concerned with linear consensus algorithms, those that can be build from matrix operations.  That is when $\phi_n$ can be represented by a matrix ${\bf \Lambda(n)}$, in the sense that
\begin{align*}
    (\phi_n f)(i) = \sum_{ij} \Lambda_{ij}(n) f(j),
\end{align*}
we will write ${\bf \Lambda(n)}f$ in place of $\phi_n(f)$.  Moreover for brevity for exposition, our focus will be on the case that our dynamics are time homogeneous in the sense that ${\bf \Lambda(n)} = {\bf \Lambda}$.

\begin{definition}(Column Stochastic Matrix) A real $N\times N$ matrix ${ \bf P}=[p_{ji}]$ is called a column stochastic matrix if $ p_{ji} \geq 0$ for $1\leq i,j\leq N$ and $\sum_{j=1}^{N}p_{ji}=1$ for $1\leq i\leq N.$ 
\end{definition}

\begin{definition}(Row Stochastic Matrix) A real $N \times N$ matrix ${ \bf A}=[a_{ij}]$ is called a row stochastic matrix if $1 \geq a_{ij} \geq 0$ for $1\leq i,j\leq N$ and $\sum_{j=1}^{N}a_{ij}=1$ for $1\leq i\leq N.$ 
\end{definition}


\begin{definition}(Irreducible Matrix)
	A $N \times N$ matrix $A$ is said to be irreducible if for any $1 \leq i, j \leq N$, there exist $m \in \mathbb{N}$ such that $(\textbf{A}^m)(i,j) > 0$, that is, it is possible to reach any state from any other state in a finite number of hops.
\end{definition}

\begin{definition}(Primitive Matrix) A non negative matrix ${\bf A}$ is primitive if it is irreducible and has only one eigenvalue of maximum modulus. 
\end{definition}

As a notational convention, matrices will be written in bold face as above.

\begin{definition}(Convex hull) \label{def: convex hull}
   For a set $U \subseteq W$, the convex hull of $U$ is the smallest convex set containing $U$,
    \begin{equation}
    co(U) = \bigcap_{\{F \hbox{ convex }: U \subseteq F \}} F.
    \end{equation}
\end{definition}

The topological closure of a set $U$, will be denoted $$\bar{U} = \bigcap_{\{F \hbox{ closed }: U \subseteq F \}} F,$$ the closure of a convex hull, will be denote $\overline{co}(U) \coloneqq \overline{ co(U) }$.

For $f \in W^{\mathcal{V}}$ we consider $co(f)$ to be the convex hull of $f$, when considered as a set of elements of $W$ indexed by $\mathcal{V}$.  More explicitly if we denote the simplex by
$
    \mathcal{S}_n \coloneqq \left\{ t \in \mathbb{R}^n : t_i \geq 0, \sum_{j=1}^n t_i = 1 \right\}
$
then for $f \in W^{\mathcal{V}},$ 
\begin{align}\label{eq: explicit convex hull of a state}
    co(f) \coloneqq \left\{ x \in W : x = \sum_{i \in V} t_i f(i), \hbox{ for } t \in \mathcal{S}_{|\mathcal{V}|} \right\}
\end{align}

\begin{definition}(Extreme point) \label{def: Extreme point}
For a convex set $U \subseteq W$ define $u \in U$ to be an extreme point of $U$, denoted $\mathscr{E}(U)$, if $u = \frac{u_1 + u_2} 2$ for $u_i \in U$ implies $u_1=u_2=u$.  For a general $U$, define $\mathscr{E}(U) \coloneqq \mathscr{E}(co(U))$.
\end{definition}
We will also have use for the following 

\begin{definition}
For a norm $\|\cdot \|$ and a set $U \subseteq W$ define the diameter of $U$ with respect to the norm $\|\cdot \|$,
$diam_{\|\cdot \|}(U) = \sup_{x,y \in U} \|x-y\|$.
\end{definition}

\begin{definition}(Convex Decreasing)
A sequence of sets $S_n \subseteq W$ is convex decreasing if $co(S_{n+1}) \subseteq co(S_{n}) $.  An $W$ consensus algorithm is convex decreasing when the sets $S_n = \Phi_n(f)$ (with $\Phi_n(f)$ defined as in \eqref{eq: phi_n update}, and the convex hull of an element of $W^{\mathcal{V}}$ defined as in \eqref{eq: explicit convex hull of a state}) are convex decreasing for any $f \in W^{\mathcal{V}}$.
\end{definition}



Our primary interest is in the case that $W = \mathbb{R}^d$ and for this case
we now recall a standard tool from Convex Geometry, the support function of a set, which we can use to give an analytic description of a Convex Decreasing sequence of sets.   For $x,y \in \mathbb{R}^n$, we use the notation $\langle x,y \rangle = x^T y = \sum_{i=1}^n x_i y_i$, where we use $(\cdot )^T$ to denote the usual transpose operation.

\begin{definition}(Support function)
    For a non-empty set $A \subseteq \mathbb{R}^d$, define its support function
    \begin{align*}
        h_A(u) \coloneqq \sup_{x \in A} \langle x, u \rangle.
    \end{align*}
\end{definition}

\begin{prop} \label{prop: basics of support functions}
    Support functions satisfy the following:
    \begin{enumerate}
        \item $A \subseteq B$ implies $h_A \leq h_B$ \label{item: inclusion}
        \item $h_A = h_{co(A)}$\label{item: support function equals support function of hull}
        \item $h_{A} = h_{\bar{A}}$ \label{item: h invariant under closure}
        \item \label{item: inclusion reversal for closed convex hull}  $h_A \leq h_B$ implies $\overline{co}(A) \subseteq \overline{co}(B)$. 
        \item A sequence of compact sets $\{A_n\}$ is convex decreasing if and only if $h_{A_n} \geq h_{A_{n+1}}$ holds for all $n$. \label{item: convex decreasing in terms of support functions}
    \end{enumerate}
\end{prop}

\begin{proof}
Observe that $A \subseteq B$ implies, $h_A(u) = \sup_{a \in A} \langle a, u \rangle \leq \sup_{b \in B} \langle b , u \rangle = h_B(u)$ giving \eqref{item: inclusion}.  

If $x \in co(A)$ then $x = \sum_{i=1}^n t_i a_i$ for some $t \in \mathcal{S}_n$ and $a_i \in A$. Thus $\langle x, u \rangle = \sum_i t_i \langle a_i , u \rangle \leq \sum_i t_i \sup_{a \in A} \langle a, u \rangle = h_A(u)$.  Thus $h_{co(A)}(u) \leq h_{A}(u)$, while the opposite inequality follows from \eqref{item: inclusion}, giving \eqref{item: support function equals support function of hull}.

If $x \in \bar{A}$ then $x = \lim_n a_n$ for $a_n \in A$, so that
$\langle x , u \rangle = \lim_n \langle a_n, u \rangle \leq h_A(u)$. Thus $h_{\bar{A}}(u) \leq h_A(u)$.  The opposite inequality follows from \eqref{item: inclusion}, and \eqref{item: h invariant under closure} follows.

Suppose that $A$ and $B$ are closed convex sets such that $h_A \leq h_B$, and take $a \in A - B$. Then, by the hyperplane seperation theorem \cite{rockafellar1970convex}, there exists $u$ such that   $\langle a, u\rangle > \sup_{b \in B} \langle b , u \rangle = h_B(u)$.  This would be a contradiction on $h_A \leq h_B$, so we must have $A \subseteq B$.  For general $A$ and $B$, we need only recall from \eqref{item: support function equals support function of hull} and \eqref{item: h invariant under closure} that $h_{A} = h_{\overline{co}(A)}$ and $h_B = h_{\overline{co}(B)}$ and apply the previous to the closed convex hull to obtain $\overline{co}(A) \subseteq \overline{co}(B)$.  Thus \eqref{item: inclusion reversal for closed convex hull} follows.

To prove \eqref{item: convex decreasing in terms of support functions} observe that if $\{A_n\}$ is a convex decreasing sequence, by definition $co(A_{n+1}) \subseteq co(A_{n})$.  So that $h_{A_n} = h_{co(A_{n})} \geq h_{co(A_{n+1})} = h_{A_{n+1}}$.  Conversely, if $h_{A_n} \geq h_{A_{n+1}}$ then by \eqref{item: inclusion reversal for closed convex hull}, $\overline{co}(A_{n+1}) \subseteq \overline{co}(A_n)$, and since the $A_n$ are assumed compact their convex hulls are as well, and hence $co(A_{n+1}) \subseteq co(A_n)$.
\end{proof}

We will also have use for a few basic results from Convex Geometry, which we collect bellow.
\begin{lem} \label{lem: basic convex geo}
    For $K$ convex and compact $co(\mathscr{E}(K)) = K$.  
    For $A_\alpha \subseteq \mathbb{R}^d$, then $co\left( \cup_\alpha A_\alpha \right) = co \left( \cup_\alpha co(A_\alpha) \right)$
\end{lem}

\begin{proof}
The first result is standard (see \cite{rockafellar1970convex}), its infinite dimensional generalization is the Krein-Milman Theorem (see \cite{Lax2002functionalanalysis}).  For the second result, clearly $\cup_\alpha co(A_\alpha) \supseteq \cup_\alpha A_\alpha$ so that $co (\cup_\alpha co(A_\alpha)) \supseteq co(\cup_\alpha A_\alpha)$. The reverse inequality follows by fixing $\alpha'$ and observing $co(\cup_\alpha A_\alpha) \supseteq co(A_{\alpha'})$, which implies $co(\cup_\alpha A_\alpha) \supseteq \cup_\alpha co(A_\alpha)$.  Since $co(\cup_\alpha co(A_\alpha))$ is the smallest convex set containing $\cup_\alpha co(A_\alpha)$ the proof is complete.
\end{proof}
\end{subsection}

\begin{subsection}{Vector Consensus framework}\label{vectorConframework}
Here, we extend a key result from \cite{kempe2003gossip,dominguez2010coordination} where a ratio of two states was maintained to reach average consensus.
We consider the network topology to be represented by a directed graph $\mathcal{G}(\mathcal{V},\mathcal{E})$ containing $|\mathcal{V}|<\infty$ nodes and satisfies the following assumptions throughout the rest of the paper.
\begin{assump}\label{ass:StrConn}
The directed graph $\mathcal{G}(\mathcal{V},\mathcal{E})$ representing the agent interconnections is strongly-connected.
\end{assump}

\begin{assump}\label{ass:colSto}
Let ${\bf P}=[p_{ji}]$ be a primitive column stochastic matrix with digraph $\mathcal{G}(\mathcal{V}, \mathcal{E})$ with $p_{ji}>0$ if and only if $(i,j) \in \mathcal{E}$.
\end{assump}


\begin{theorem}
    A sequence of matrices ${\bf A_n}$ defines a scalar consensus algorithm if and only if it defines an $\mathbb{R}^d$ vector consensus algorithm.
\end{theorem}

\begin{proof}
The value of $i$-th node in the $l$-th coordinate after $n$ iterations is the application of $n$-iterations to the $l$-th coordinate function evaluated at the $i$-th node,  $(\prod_{k=1}^n {\bf A_k})(f)(i)_l = (\prod_{k=1}^n {\bf A_k})(f_l) (i)$.  Hence the theorem follows from the existence of $c$ and $C >0$ (dependent on dimension and choice of norm $\| \cdot \|$) such that
\begin{align*}
    c \max_l |\Phi_n(f)(i)_l &- \Phi_n(f)(j)_l | \leq
    \| \Phi_n(f)(i) - \Phi_n(f)(j) \| 
        \\
       & \leq 
            C \max_l |\Phi_n(f)(i)_l - \Phi_n(f)(j)_l |
\end{align*}
with the fact above that $\Phi_n(f_l)(i) = \Phi_n(f)(j)_l$ the result follows.
\end{proof}

Each node $i \in \mathcal{V}$ maintains three state estimates at time $k$, denoted by $x^{i}(k) \in \mathbb{R}^d$ (referred as numerator state of node $i$), $y_{i}(k) \in \mathbb{R}$ (referred as denominator state of node $i$) and $r^{i}(k) \in \mathbb{R}^d$ (referred as ratio state of node $i$). Here $d\ (\ge 1)$ is the dimension of each node's state. Node $i$ updates its numerator and denominator states at the $(k+1)^{th}$ discrete iteration according to the following update law: 
\begin{align}
x^{j}(k+1)&=\sum_{i\in N^-_j}p_{ji}x^{i}(k),\label{eq:numeratorState}\\ 
y_{j}(k+1)&=\sum_{i\in N^-_j}p_{ji}y_{i}(k),\label{eq:denominatorState}
\end{align}
\noindent where, $N^-_i$ is the set of in-neighbors of node $i$.  We will use the notation $x(k+1) = {\bf P} x(k)$ as shorthand for \eqref{eq:numeratorState}, and observe that $x(n) = {\bf P}^n x(0)$. The initial conditions for the numerator vector state and denominator state for any node $i \in \mathcal{V}$ are:
\begin{align}\label{eq:initial_condition}
    x^i(0) &= [x_1^i(0) \ x_2^i(0) \dots x_d^i(0)]^T, \ y_i(0) = 1.
\end{align} 
Node $i$ further updates its ratio state as: 
\begin{align}
r^i(k+1)&= \frac{1}{y_i(k+1)}x^i(k+1),\label{eq:ratioState}
\end{align}
Under Assumptions~\ref{ass:StrConn},~\ref{ass:colSto} and the initialization in~(\ref{eq:initial_condition}), ratio state in \eqref{eq:ratioState} is well defined.
The next theorem establishes the convergence of the ratio state, which is a direct and simple generalization of the result in \cite{kempe2003gossip,dominguez2010coordination}.

\begin{theorem}\label{thm:consensus_conv}
Let $\{x^i(k)\}, \{y_i(k)\}$ and $\{r^i(k)\}$ be the sequences generated by~(\ref{eq:numeratorState}), (\ref{eq:denominatorState}) and (\ref{eq:ratioState}) respectively. Let the initial conditions for the network states be as defined in~(\ref{eq:initial_condition}). Then, under Assumptions~\ref{ass:StrConn} and~\ref{ass:colSto} the ratio state $r^i(k)$ asymptotically converges to $\overline{r}:= \lim \limits_{k \rightarrow \infty} \frac{1}{y_i(k)}x^i(k) = \frac{1}{N}\sum\limits_{j=1}^N x^j(0)$ for all $i \in \{1,...,N\}$. 
\end{theorem}

\begin{proof}
The proof follows by applying the result from \cite{kempe2003gossip,dominguez2010coordination} applied component wise to $x^i$ and $r^i,$ and thus we have convergence in the case when node states are vectors. \end{proof}

In a slightly different framework, we consider the average consensus problem, where each node $i \in \mathcal{V}$  maintains a single state $z^i(k) \in \mathbb{R}^d$, for each time $k$, and update its state according to the following update law:
\begin{align} \label{eq:Row stochastic updates}
    z^i(k+1) = \sum_{j\in N^-_i}a_{ij} z^j(k),
\end{align}
\noindent where ${\bf A}=[a_{ij}]$ is a primitive, and row-stochastic
, with $a_{ij} > 0$ if and only if $(i,j) \in \mathcal{E}$. In this case, $z^i(k)$ converges independent of $i$, to $\sum\limits_{j=1}^N \pi_j z^j(0)$ for some $\pi \in \mathcal{S}_n$ (see \cite{levin2017markov}), and $\pi_i = \frac 1 N$ in the case that ${\bf A}$ is assumed to be column stochastic as well (see \cite{olfati2007consensus}).
\end{subsection}

\section{Convex Hull based Finite-Time stopping Criterion}\label{sec:ConHull}

The following theorem shows that consensus algorithms that are convex decreasing converge to the same finite limit at all nodes, and that if one sets an threshold for convergence with an open set about the consensus, the threshold will be met in finite time.  Further, when the threshold set is assumed convex, it is proven that if all agents possess a value within the set, their updated values remain within this threshold.  In this sense, convex threshold sets provide a guarantee on future behavior of the network.

\begin{theorem} \label{thm: Convex decreasing consensus algorithms have good convergence properties}
    Suppose that $S_n = \Phi_n(f)$ represents the state of a convex decreasing consensus algorithm, then $S_\infty \coloneqq \lim_n \Phi_n(f)_i$ is finite and well defined independent of $i$.  Further, given a set with non-empty interior $\mathcal O$, containing $S_\infty$, there exists $n_0$ such that $n \geq n_0$ implies $S_n \subseteq \mathcal O$.  If $K$ is a convex set such that $S_M \subseteq K$, then $S_n \subseteq K$ for $n \geq M$.
\end{theorem}

\begin{proof}
    For a convex set $K$, if $S_M \subseteq K$, then $co(S_M) \subseteq K$, and since $co(S_n)$ are nested, the last statement follows immediately.  Note that by Cantor's intersection theorem (see for example \cite[Lemma 3.2.2]{gray2009probability}), since $co(S_n)$ are nested, compact (since the convex hull of finitely many points, $\{\Phi_n(f)_i\}_{i=1}^{|\mathcal{V}|}$, can be expressed as the continuous image of the simplex, a compact set) sets, $\cap_n co(S_n)$ is non-empty. Since the mapping $(x,y) \mapsto \| x - y\|$ is a convex map\footnote{Indeed, the inequality $\|(1-t)x_0 + tx_1  - (1-t)y_0 - t y_1 \| \leq (1-t) \|x_0 -y_0\| + t \|x_1 - y_1\|$ follows from an application of the triangle inequality and scalar homogeneity, for any $t \in (0,1)$ and $x_i, y_i \in W$.}, it follows that
    \begin{align*}
        diam(co(S_n)) = \max_{ij} \| \Phi_n(f)_i - \Phi_n(f)_j \|.
    \end{align*}
    Thus the diameter of $co(S_n)$ is the maximum of finitely many terms tending to zero, and hence $\lim_n diam(co(S_n)) = 0$, and the non-empty set $\cap_n co(S_n)$ can contain at most one point, which we denote $S_\infty$.\\
    Given $\mathcal{O}$, an open set containing $S_\infty$, $co(S_n) \subseteq \mathcal{O}$ for large enough $n$, since $diam(co(S_n))$ tends to zero and $\mathcal{O}$ contains an $\varepsilon$ ball about $S_\infty$ with respect to $\| \cdot \|$ for small enough $\varepsilon$, since all finite dimensional norms are equivalent.
\end{proof}

Theorem \ref{thm: Convex decreasing consensus algorithms have good convergence properties} allows us to provide stopping guarantees to convex decreasing consensus algorithms, particularly useful in distributed contexts.  We will show that 
in both consensus frameworks \eqref{eq:ratioState} and \eqref{eq:Row stochastic updates} the network states $\{ r^i(k)\}_{i=1}^N$ and $\{z^i(k)\}_{i=1}^N$ at time $k$ define a sequence of polytopes $\{r(k)\}_{k=0}^\infty$ and $\{z(k)\}_{k=0}^\infty$ respectively defined to be
\begin{align*}
    r(k) 
        &\coloneqq 
            co(\{r^i(k)\}_{i=1}^N)
        \\
    z(k) 
        &\coloneqq 
            co(\{z^i(k)\}_{i=1}^N )
\end{align*}
that are convex decreasing.

\begin{theorem} \label{thm: examples of convex decreasing}
Consider the update equation \eqref{eq:Row stochastic updates} for $z(k)$ where $z(k+1)={\bf A}z(k)$ with the assumption that ${\bf A}$ is row stochastic and the update equations for $r(k)$ as given in \cite{fax2002information,olfati2007consensus,nedic2014distributed,khatana2019gradient}, where $x(k+1)={\bf P}x(k); y(k+1)={\bf P} y(k)$ and $r(k+1)=\frac{1}{y(k)}x(k)$ where ${\bf P}$ is column stochastic. Then $z(k)$ and $r(k)$ form convex decreasing consensus algorithms. 
\end{theorem}

\begin{proof}
    To see that $z(k)$ is convex decreasing is immediate, since by definition 
    \begin{align*}
       z^i(k+1) =\sum_{j\in N^-_i}a_{ij} z^j(k),
    \end{align*}
    where $\{a_{ij} \}_j$ is a sequence of non-negative numbers that sum to one.  Thus $z^i(k+1)$ is a convex combination of elements of $z(k)$ and hence $z(k+1) = co(\{ z^i(k+1) \}) \subseteq z(k)$.\\
    To see that $r(k)$ are convex decreasing, since $r(k)$ is finite and hence compact for all $k$, by Proposition \ref{prop: basics of support functions}, it is enough to show that their support functions are decreasing, that is $h_{r(k+1)} \leq h_{r(k)}$. Note that from $r^j(k) = x^j(k)/y^j(k)$ the support function satisfies the following inequality for all $j$,
    \begin{align*}
        \langle x^j(k), u \rangle \leq h_{r(k)}(u) y^j(k).
    \end{align*}
    With ratio-consensus updates from ${\bf P}$ column stochastic, to prove convex decreasingness, it suffices to show that $\langle r^j(k+1), u \rangle \leq h_{r(k)}(u)$, or equivalently, 
    \begin{align*}\langle x^j(k+1), u \rangle \leq h_{r(k)}(u) y^j(k+1).
    \end{align*}
    Computing,
    \begin{align*}
        \langle x^j(k+1), u \rangle 
            &=
                \sum_i p_{ji} \langle x^i(k), u \rangle 
                    \\
            &\leq
                \sum_i p_{ji} h_{r(k)} y^i(k)
                    \\
            &=
                 h_{r(k)}(u) y^j(k+1).
    \end{align*}
    That $r(k)$ and $z(k)$ are consensus algorithms follows from well known literature. In particular $r(k)$ is a consensus algorithm by Theorem \ref{thm:consensus_conv}; indeed, $r(k)$ converges to the average $\sum_{i=1}^n x^{i}(0) \in \mathbb{R}^d.$.  For $z$, observe that $z(n) = {\bf P}^n z(0)$.  The connectivity properties of ${ \bf P}$, ensure that
    \begin{align*}
        \lim_{n \to \infty} {\bf P}^n = \left( \begin{array}{cccc}
            \pi_1 & \pi_2 & \dots & \pi_{|\mathcal{V}|} \\
            \pi_1 & \pi_2 & \dots & \pi_{|\mathcal{V}|} \\
            \vdots & \vdots & \vdots & \vdots \\
            \pi_1 & \pi_2 & \dots & \pi_{|\mathcal{V}|} \\
            \end{array}
            \right)
    \end{align*}
    for a $\pi \in \mathcal{S}_{|\mathcal{V}|}$ (see \cite{levin2017markov} for example, and note in the language of Markov Chains that $p_{ii} >0$ ensures aperiodicity, while irreducibility follows from strong connectedness).  As a consequence $\lim_n z(n) = \lim_n {\bf P}^n z(0) = \left( \begin{array}{c} 
            \sum_i \pi_i z_i(0) \\
            \vdots \\
            \sum_i \pi_i z_i(0).
            \end{array} \right)$.  
\end{proof}

Below in Figure \ref{fig:convex hull}, we briefly illustrate the result in Thoerem~\ref{thm: examples of convex decreasing} via an example, for the update $z(k+1)={\bf A} z(k)$ with ${\bf A}$ being row stochastic a 30 node Erdos-R\'enyi graph is initialized with values chosen uniformly at random from $(0,1)^2$.  The initial values are displayed in red, with the boundary of the convex hull traced in black.  
The points of $z(0)$ are displayed in red, $z(1)$ in blue, $z(2)$ in green, with the convex hull boundaries of the respective sets traced in black.  That row stochastic updating is convex decreasing is instantiated in the nested-ness of the sets $z(i)$.

\begin{wrapfigure}{r}{0.5\textwidth}
    \begin{center}
      \includegraphics[width=0.48\textwidth]{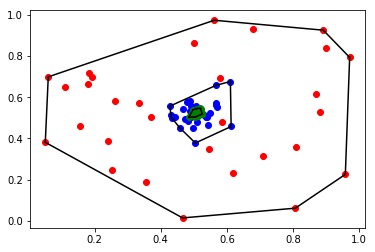}
      \caption{Iterations of a convex decreasing algorithm}
      \label{fig:convex hull}
    \end{center}
  \end{wrapfigure}

When $d =1$, the convex hull is simply described, $co(r(k)) = [\min_j \{r^j(k) \}_{j=1}^n, \max_j \{ r^j(k) \}_{j=1}^n]$, so that $co(r(k)) \supseteq co(r(k+1))$ gives the monotonicity results from \cite{prakash2019distributed}, $\min_j \{r^j(k) \}_{j=1}^n, \leq \min_j \{ r^j(k+1) \}_{j=1}^n$ and $\max_j \{r^j(k) \}_{j=1}^n, \geq \max_j \{ r^j(k+1) \}_{j=1}^n$.  Analogously, applying Theorem \ref{thm: examples of convex decreasing} in the one dimensional case delivers the monotonicity of min and max from \cite{yadav2007distributed}. In the $d$-dimensional case, taking $\mathcal{F}$ to be the family of all rectangular sets recovers Theorem $4$ of \cite{khatana2019gradient}.

We will use the monotonicity of convex decreasing consensus algorithms to develop a distributed stopping criteria, guaranteeing convergence of all nodes within an $\epsilon$-ball of the consensus value for a general norm.  
First we develop a distributed convex hull computation algorithm that is of independent interest.  Here, given a function $w: {\mathcal{V}} \to \mathbb{R}^d$ the convex hull $co(w) \coloneqq co(\{w(i)\}_{i=1}^n)$ is to be determined by the network in a distributed fashion.  Consider a algorithm where at stage one, agents share their value $w(j)$ with neighbors.  The nodes then update their approximation of the convex hull, by determining the extreme points among all values received from their neighbors and their own.  This new set of extreme points is communicated to all neighbors and then the process repeats. We show that after $D$ ($D$ being the diameter of the network) iterations, every node will have determined the extreme points of $w$.

In the context of a consensus algorithm protocol $w: {\mathcal{V}} \to \mathbb{R}^d$ represents $w(i) = r^i(k)$ or $z^i(k)$ the value at node $i$ in iteration $k$ of a consensus algorithm, and we implement the following stopping criterion.  Given a norm $\| \cdot \|$ and a tolerance $\varepsilon$, implement the convex hull algorithm at time $k$, then at time $k+D$, at a node $j$ if $\max_{e(i),e(j) \in \mathscr{E}(w)} \| e(i) - e(j) \| \leq \varepsilon$, then stop the consensus algorithm.  In what follows we demonstrate that upon stopping every node is within $\varepsilon$ of the consensus value in norm.

\section{Peer to Peer Convex Hull Algorithm}

We now describe a finite time algorithm for distributed convex hull computation.  Suppose that $|\mathcal{V}|$ agents indexed by $i$  where each agent has a set $S_i$ of elements of $\mathbb{R}^d$.  The agent can communicate with each other while respecting constraints imposed by a specific communication network. We provide a distributed consensus algorithm through which 
all agents obtain $\mathscr{E} \left( \cup_i S_i \right)$ in $D$-iterations of the algorithm.  If we let $E$ denote the space of all finite sequences of elements of $\mathbb{R}^d$, the convex hull algorithm can be understood as a distributed consensus algorithm.

\begin{definition} \label{def: convex hull algorithm}
    For $S: \mathcal{V} \to E$, define the initialization $x_i(0) = \mathscr{E}(S_i)$.  Iteratively define,
    \begin{align*}
     x_i(t) &= \mathscr{E}(\bigcup_{j \in N_{i}^-(1)} x_j(t-1)).
    \end{align*}
   We identify $x_i(t)$ with an element of $E$, by writing its elements in lexicographical order.
\end{definition}

That is, at iteration $t$, an agent $i$ receives the extreme points known to their ``in-neighbors'' and forms a new set $s_i(t)$ comprised of their previous extreme points and their neighbors.  Agent $i$ finds the extreme points of this new set, and then communicates the set to its ``out-neighbors'' to initiate another iteration. 

\begin{theorem}
For any initial configuration defined by $S: {\mathcal{V}} \mapsto E$, the algorithm for $x_i(t)$ described in Definition \ref{def: convex hull algorithm} is a distributed consensus algorithm.  Moreover, considered as sets
\begin{align} \label{eq: consensus of hull algorithm}
    x_i(t) = \mathscr{E}(S_i(t))
\end{align}
 where $S_i(t) \coloneqq  \cup_{j \in N_i^-(t)}S_j$, and we recall that $N_i^-(t)$ is the $t$ in-neighborhood of node $i$ (see Definition~\ref{def: in-neighborhood}).  In particular, for $t \geq D$, $x_i(t) = x_i(D) = \mathscr{E}( \cup_{i \in \mathcal{V}} S_i)$.
\end{theorem}

\begin{proof}
The result is true by definition checking when $t=0$, since $S_i(0) = S_i$. Thus we proceed by induction and assume the result holds for $k< t$.
By definition,   
\begin{align}
        x_i(t)
            =
                \mathscr{E}\left( \bigcup_{j \in N_i^-(1)} x_j(t-1) \right) \label{eq: cvx hull eq 1}.
\end{align}
By the induction hypothesis,
\begin{align*}
        \mathscr{E}\left( \bigcup_{j \in N_i^-(1)} x_j(t-1) \right)
            &=
                \mathscr{E}\left( \bigcup_{j \in N_i^-(1)} \mathscr{E}(S_j(t-1)) \right)
                    \\
            &=
                \mathscr{E}\left( \bigcup_{j \in N_i^-(1)} \mathscr{E}\left( \bigcup_{k \in N_j^-(t-1)} S_k \right) \right).
    \end{align*}
    Recall that for non-convex sets $U$, $\mathscr{E}(U) \coloneqq \mathscr{E}(co(U))$ , so that 
    \begin{align*}
        \mathscr{E}\left( \cup_{j } \mathscr{E} \left( \cup_{k_j } S_{k_j} \right) \right)
            =
                \mathscr{E}\left( co ( \cup_{j } \mathscr{E} \left( co( \cup_{k_j } S_{k_j} ) \right) ) \right)
    \end{align*} 
    where the subscript $k_j$ ranges over the set $N_j^-(t-1)$.
    If we write $K_j = co( \cup_{k_j}S_{k_j} )$, and apply Lemma \ref{lem: basic convex geo}, with the fact that $K$ is convex and compact,
    \begin{align*}
        co ( \cup_{j } \mathscr{E} \left(  K_j \right) ) 
            &=
                 co ( \cup_{j } co( \mathscr{E} (K_j)))
                    \\
            &=
                co ( \cup_{j } K_j ).
    \end{align*}
    By definition of $K_j$ and another application of Lemma \ref{lem: basic convex geo},
    \begin{align*}
        co ( \cup_{j } K_j )
            &=
                 co ( \cup_{j } co( \cup_{k_j } S_{k_j} )  ) 
                    \\
            &=
                co ( \cup_{j }  \cup_{k_j } S_{k_j} ) .
    \end{align*}
   Thus our result follows once we can show,
   \begin{align*}
        \bigcup_{j \in N_i^-(1)} \left( \bigcup_{k \in N_j^-(t-1)} S_k \right)  = \bigcup_{j \in N_i^-(t)} S_j.
   \end{align*}
    Both sets can be considered as unions of $S_k$ indexed by paths of length not larger than $t$ terminating at $i$.  More explicitly, both sets can be written as
    $
        \bigcup_{\lambda \in \Lambda} S_\lambda
   $
    where $\Lambda$ is the space of all paths $v:\{0,1,2, \dots, k \} \to V$ such that $k \leq t$, $v(k) = v_i$.   This gives \eqref{eq: consensus of hull algorithm}.  Since $N_i^-(t) = V$ for $t \geq D$, $S_i(t) = \cup_{j \in \mathcal{V}}S_j$ and $x_i(D) = x_j(D) = \mathscr{E}(S)$.  Thus the algorithm considers is a consensus algorithm. The algorithm is distributed as each $S_i(t)$ is a function of the $S_j(t-1)$ for $j \in N_i^-(1)$.
\end{proof}

This shows that, agents in a distributed network can obtain exact knowledge of the convex hull in $D$ iterations.  As an application the convex hull algorithm can be used to provide finite time stopping criterion for a convex decreasing consensus algorithm. We need the following lemma.
\begin{lem}
For a norm $\|\cdot\|$ and a convex set $K$,
\begin{align*}
    diam_{\|\cdot \|}(K) = \sup_{w_1, w_2 \in \mathscr{E}(K)} \| w_1 - w_2\|.
\end{align*}
\end{lem}

\begin{proof}
For fixed $y \in K$ that, $x \mapsto \| x-y\|$ is convex and hence takes its maximum value on $K$ at extreme value of $K$.  Hence $\sup_{x \in K} \| x - y\| = \sup_{w_1 \in \mathscr{E}(K)} \| w_1 - y \|$, applying the same argument again we obtain 
\begin{align*}
    \sup_{w_1 \in \mathscr{E}(K)} \| w_1 - y \| 
        &=
            \sup_{w_1 \in \mathscr{E}(K)} \sup_{w_2 \in \mathscr{E}(K)} \| w_1 - w_2 \|,
\end{align*}
and our result follows.
\end{proof}

\begin{theorem}
    If $c^i(k)$ denotes the vector at node $i$ at time $k$ in a convex-decreasing consensus algorithm, then for $k' \geq k$
    \begin{align*}
    \|c^i(k') - \lim_{n \to \infty} c^i(n)\| \leq \max_{w_1,w_2 \in \mathscr{E}(c^k)} \|w_1 - w_2\|.
    \end{align*}
\end{theorem}

\begin{proof}
    Denoting by $c(k)$ the element of $(\mathbb{R}^d)^{\mathcal{V}}$ defined by $i \mapsto c_i(k)$, the assumption that $c(k)$ is convex decreasing implies $co(c(k')) \subseteq co(c(k))$ for $k' \geq k$, and hence $c_i(k') \in co(c(k))$ for all $i$.  Thus, $\lim_n c^i(n) \in co(c(k))$ as well and we have
        \begin{align*}
            \| c^i(k') - \lim_n c^i(n) \|
                &\leq
                    diam_{\| \cdot \|} (c(k))
                        \\
                &=
                    \max_{w_1,w_2 \in \mathscr{E}(c_k)} \|w_1 - w_2\|. 
        \end{align*}
\end{proof}

It follows that an agent can obtain exact bounds on the distance from convergence of the consensus with respect to an arbitrary norm.

Standard algorithms for computing the convex hull of a set of points in $d$-dimensional exist, see \cite{clarkson1989applications,barber1996quickhull}.   However such can easily be prohibitively expensive especially in high dimension (worst case runtime is of the order $O(n^{d/2})$), when computational resources, or communication power is limited.  Further, in the worst case scenario, the number of extreme points can be of the same order as the nodes of the graph (take $w(i)$ to be points of a $d$-dimensional sphere for instance), and hence their communication cost is equivalent to that of the entire system state.  In the following section we develop a stopping algorithm to address these potential feasibility issues.

\section{Norm Based Finite-Time Termination}\label{sec:normBase}

Similar to the convex hull comprising all points (corresponding to each agent), radius of a minimal ball in $d$ dimension enclosing all the points can also be used as a termination criterion. Once the radius is within some bound $\rho$, it can be shown that every agent's state is within $2\rho$ of the consensus value. We remark that even in the $p-norm$ case determination of a minimum norm ball in a distributed manner is a difficult problem (see \cite{fischer1975smallest}). Here, we provide an algorithm which distributedly finds an approximation of minimal ball at each agent. We next show that the minimal ball is enclosed in this approximation, thus if the approximate ball's radius is within $\rho$ then the minimal ball's radius is within $\rho$ as well. This is established in next Lemma.

\begin{lem}\label{lem:RadiusUpdate} Let $\{c^i(k)\}$ be the sequence generated by a distributed convex-decreasing consensus protocol. 
For all $i \in \V$, let
\begin{align}\label{eq:radiusUpdate}
    R_i(k+1,k') \coloneqq \underset{j\in  N^-_i}{\max} &\{\|c^i(k'+k+1)-c^j(k'+k) \|\\
    &+ R_j(k,k')\}
\end{align}
with $R_i(0,k'):=0$ and $k'\geq0$. Then
\begin{equation}\label{eq:radiusBall}
    c^j(k')\in B\{R_i(D,k'),c^i(k'+D)\},
\end{equation}
for all $j \in \mathcal{V}$,
where $B\{R,x\}$ denotes the closed ball of radius $R$ centered at $x$ and $D$ is the diameter of the underlying graph topology.
\end{lem}
\begin{proof}
 We first prove the following claim.
\begin{equation}\label{eq:radiusBallClaim}
    c^j(k')\in B\{R_i(k,k'),c^i(k'+k)\}
\end{equation}
for all $j \in \V$ such that the length of the shortest path from $i$ to $j$ is less than equal to $k.$ Clearly above claim is sufficient to prove (\ref{eq:radiusBall}) as when $k=D$, (\ref{eq:radiusBallClaim}) is valid for all $j \in \V$. Let the length of the shortest path from $i$ to $j$ be denoted as $|path(i,j)|$. 
We prove the claim using induction. For $k=1$,
\begin{equation*}
     R_i(1,k')=\underset{j\in  N^-_i}{\max} \|c^i(k'+1)-c^j(k')\|.
\end{equation*}
Then for all $j\in  N^-_i$, that is for all $j$ such that $|path(i,j)|\leq 1$, we get
\begin{align*}
     \|^i(k'+1)-c^j(k')\|& \leq R_i(1,k'), \mbox{ and thus }\\
     \implies c^j(k')\in &B\{R_i(1,k'),c^i(k'+1)\}.
\end{align*}
Thus the assertion holds for $k=1$. Now lets assume (\ref{eq:radiusBallClaim}) is true for $k$. Let $j$ be a node such that $|path(i,j)|\leq k+1$. Let $q$ be a neighbor of $i$ on the shortest path from $i$ to $j$, then $|path(q,j)|\leq k$. Then from induction assumption,
\begin{equation*}
    c^j(k')\in B\{R_q(k,k'),c^q(k+k')\},
\end{equation*}
that is,
\begin{equation}\label{eq:radiusBallInd}
     \|c^q(k'+k)-c^j(k')\| \leq R_q(k,k').
\end{equation}
From definition of $R_i(k+1,k') $,
\begin{equation*}
    \|c^i(k'+k+1)-c^q(k'+k)\| + R_q(k,k') \leq R_i(k+1,k') .
\end{equation*}
From triangle inequality,
\begin{multline*}
     \|c^i(k'+k+1)-c^j(k')\| \leq \|c^i(k'+k+1)-c^q(k'+k)\| \\
     + \|c^q(k'+k)-c^j(k')\|.
\end{multline*}
Using (\ref{eq:radiusBallInd}),
\begin{multline*}
     \|c^i(k'+k+1)-c^j(k')\| \leq \|c^i(k'+k+1)-c^q(k'+k)\| \\
     + R_q(k,k')
\end{multline*}
which implies that 
\begin{align*}
     \|c^i(k'+k+1)-c^j(k')\| \leq R_i(k+1,k')
\end{align*}
and thus,
\begin{align*}
    c^j(k')\in B\{R_i(k+1,k'),c^i(k'+k+1)\},
\end{align*}
and the result follows.
\end{proof}

Lemma~\ref{lem:RadiusUpdate} provides a distributed way to find a ball which encloses all the nodes. Only information needed by a node is the current radius of its neighbors (along with the states pertaining to ratio consensus) and it can determine the final radius within $D$ iteration. Further, since the ball $B\{R_i(D,k'),c^i(k'+D)\}$ encloses all the nodes, it also encloses the minimum ball, as mentioned earlier. Thus we have provided an algorithm to find an approximation of the minimum ball comprising of all nodes. We next present a framework which we use to prove that this radius converges to $0$ and can be used as a distributed stopping criterion.

Consider the coordinate-wise maximum and minimum of the states taken over all the agents at atime instant $k$ be given by, $M(k) = [M_1(k)\ M_2(k)\ \dots M_d(k)]^T$ and $m(k) = [m_1(k)\ m_2(k)\ \dots m_d(k)]^T$ respectively. That is, \begin{align}
    M_s(k)&\coloneqq\underset{i\in \mathcal{V}}{\max} \ c^i_s(k) \label{eq:maxvec}\\
  m_s(k)&\coloneqq\underset{i\in \mathcal{V}}{\min} \ c^i_s(k) \label{eq:minvec} 
\end{align}where $M_s(k) \in \mathbb{R}$, $m_s(k) \in \mathbb{R}$ for all $s \in \{1,2,\dots, d\}$ and $c_s^i(k)$ is the $s$-th elements of $c^i(k)$.
Then from \cite{ratio_consensus_lab}, for all time instants $k^{'}\geq k$ and for all $i\in \mathcal{V}$ and $s \in \{1,2,\dots, d\}$,
\begin{align}\label{eq:maxMinMono}
    m_s(k) \leq c^{i}_s(k^{'}) \leq M_s(k).
\end{align}
Further from \cite{ratio_consensus_lab}, for all $i\in \mathcal{V},\ l\geq 0$ and $s \in \{1,2,\dots, d\}$,
\begin{align}
     M_s((l+1)D) &< M_s(lD)\label{eq:maxStrict}\\
    m_s((l+1)D) &> m_s(lD).\label{eq:minStrict}
\end{align}
By using (\ref{eq:maxMinMono}), (\ref{eq:maxStrict}) and (\ref{eq:minStrict}), we can prove the following theorem.
\begin{theorem}\label{thm:mxpMNPCvg}
Consider $c$ to a convex decreasing consensus algorithm
\begin{align*}
     \lim_{l\to\infty}  M(lD)=\lim_{l\to\infty}  m(lD) = c^\infty \coloneqq \lim_{l \to \infty} c^i(l).
\end{align*}
where, $M(k)$ and $m(k)$ are as defined in \eqref{eq:maxvec} and \eqref{eq:minvec}.
\end{theorem}
\begin{proof}
It follows from Theorem~\ref{thm: Convex decreasing consensus algorithms have good convergence properties} that $\lim\limits_{k\rightarrow \infty} c^i(k) = c_\infty$  for all $i \in \mathcal{V}$. This implies that for all $i \in \mathcal{V}$, $s \in  \{1,2,\dots, d\}$ and given $\epsilon >0$, there exists a $L>0$ such that for all $k \geq L, \|c^i_s(k)-c^\infty_s \|< \epsilon$. This implies that, there exists $Q>0$ such that for all $k\geq Q,\ \|\max\limits_{i} c^i_s(k)-c^\infty_s\| < \epsilon$. Similarly, $\|\min\limits_{i} c^i_s(k)-c^\infty_s\| < \epsilon$. Thus, it follows that, $\lim\limits_{k \rightarrow \infty} M_s(k) = c^\infty_s$ and $\lim\limits_{k \rightarrow \infty} m_s(k) = c^\infty_s$.  As subsequences of convergent subsequences, the same conclusion follows for $M(lD)$ and $m(lD)$.  \end{proof}

 
\begin{corollary}\label{cor:mxpMNP}
For $c$ a convex decreasing consensus algorithm,  then, 
\begin{align} \label{eq: min minus max to zero}
    \lim\limits_{l \rightarrow \infty}||M(lD)-m(lD)|| = 0.
\end{align}
In particular \eqref{eq: min minus max to zero} holds for the ratio consensus protocol of (\ref{eq:numeratorState}), (\ref{eq:denominatorState}) and (\ref{eq:ratioState}) when Assumptions ~\ref{ass:StrConn} and~\ref{ass:colSto} hold.
 \end{corollary}
 
\begin{proof}
The proof directly results from Theorem~\ref{thm:mxpMNPCvg}
. \end{proof}

It is clear from Lemma~\ref{lem:RadiusUpdate} that at any instant $k$, all agents' states are within $2R_i(D,k)$ of each other, that is,
\begin{equation}
    \underset{i,j\in \V}{\max} \|c^i(k)-c^j(k)\| \leq 2R_i(D,k).
\end{equation}
Thus if $R_i(D,k)$ is within a tolerance $\rho/2$, all the agents ratio state will be within $\rho$ of consensus. We next provide convergence result for $R_i(D,k)$ as $k \to \infty$.
\begin{theorem}\label{thm:radiusConvg}
For a distributed convex decreasing consensus algorithm $c$ 
and update as in (\ref{eq:radiusUpdate}). Let $\overline{R_i}(l):=R_i(D,lD)$ for $l=0,1,2,\ldots$ and all $i \in \V.$ Then
\begin{align*}
    \lim\limits_{l \rightarrow \infty}\overline{R_i}(l)&= 0
\end{align*}
for all $i \in \V.$
\end{theorem}

\begin{proof}
From definition,
\begin{equation*}
R_i(1,lD)=\underset{j\in  N^-_i}{\max}\{\|c^i(1+lD)- c^j(lD)\|+ R_j(0,lD)\}
\end{equation*}
which implies,
\begin{multline}\label{eq:rL1}
R_i(1,lD)\leq\underset{j\in  N^-_i}{\max}\|c^i(1+lD)-c^j(lD)\|\\
+\underset{j\in  N^-_i}{\max} R_j(0,lD)
\end{multline}
Let $M(lD)$ and $m(lD)$ be as defined in Theorem~\ref{thm:mxpMNPCvg}. 
As $c^i(1+lD)\geq m(lD)$ and $c^j(lD) \leq M(lD)$ from (\ref{eq:maxMinMono}), we get
\begin{equation}\label{eq:rL2}
    \|c^i(1+lD)- c^j(lD)\| \leq \|M(lD)-m(lD)\|
\end{equation}
Then using  (\ref{eq:rL2}) in (\ref{eq:rL1}) and observing $R_j(0,lD)=0$ for all $j \in \V$, we get
\begin{equation}\label{eq:rL3}
R_i(1,lD)\leq  \|M(lD)-m(lD)\|.
\end{equation}
Similarly,
\begin{multline}\label{eq:rL4}
R_i(2,lD)\leq\underset{j\in  N^-_i}{\max}\|c^i(2+lD)-c^j(1+lD)\|+\\
    \underset{j\in  N^-_i}{\max} R_j(1,lD)
\end{multline}
Again as $c^i(2+lD)\geq m(lD)$ and $c^j(1+lD) \leq M(lD)$, we have
\begin{equation}\label{eq:rL5}
    \|c^i(2+lD)- c^j(1+lD)\| \leq \|M(lD)-m(lD)\|
\end{equation}
Then using  (\ref{eq:rL3}), (\ref{eq:rL4}) and (\ref{eq:rL5}), we get
\begin{equation*}
R_i(2,lD)\leq  2\|M(lD)-m(lD)\|.
\end{equation*}
Following the same process, we have
\begin{equation}
\overline{R_i}(l) = R_i(D,lD)\leq  D\|M(lD)-m(lD)\|.
\end{equation}
Then from Corollary~\ref{cor:mxpMNP},
\begin{equation*}
    \lim\limits_{l \rightarrow \infty}\overline{R_i}(l)= 0 
\end{equation*}
\end{proof}

Notice that $R_i(l)$ can be different for different nodes and each node might detect $\rho$-convergence ($R_i(l)<\rho$) at different time instants. According to Lemma~\ref{lem:RadiusUpdate}, once $R_i(l)<\rho$ for any $i \in \V$, $\|c^i(lD)-c^j(lD)\|<2\rho$, that is the ratio state is within $2\rho$ of consensus value, and the consensus is achieved. Further, any node $i$ which detects convergence can propagate a ``converged flag" in the network. To take that into account, we run a separate $1$-bit consensus algorithm (denoted as convergence consensus) for each node where each node maintains a convergence state $b_i(k)$ and shares it with neighbors. Each node initializes $b_i(k)$ at every $lD$ iteration for $l \in \{0,1,2,\dots\}$ with $1$ or $0$ depending on the node has detected convergence or not, and updates its value on every iteration using,
\begin{align}\label{eq:bitCons}
    b_i(k+1)= \underset{j\in N^-_i}{\bigcup}b_j(k),
\end{align}
\noindent where $\bigcup$ denotes OR operation, $k\geq 0$ and $b_j(0)=1$ if node $j$ has detected convergence at initialization instant $0$ and $b_j(0)=0$ otherwise. Clearly, if $b_j(0)=1$ for any $j\in\V$, then $ b_i(D)=1$ for all $i\in\V$ where $D$ is the diameter. Thus each node can use $b_i(D)$ as a stopping criterion.

Using above discussion and Theorem~\ref{thm:radiusConvg}, we present an algorithm (see Algorithm~\ref{alg:radiusAlg}) instantiating the result for ratio consensus (which could easily be adapted for more general settings), which determines the radius $\overline{R_i}(l)$ for $l=0,1,2,\ldots$ and all $i \in \V$  and provides a finite-time stopping criterion for vector consensus.

\begin{theorem}
Algorithm~\ref{alg:radiusAlg} converges in finite-time simultaneously at each node.
\end{theorem}
\begin{proof}
From Corollary~\ref{cor:mxpMNP}, it follows that $\overline{R_i}(l)\rightarrow0$ as $l\rightarrow\infty.$ Thus, for any given $\rho>0$ and node $i \in \V$ there exists an integer $t(\rho,i)$ such that for $l= t(\rho,i),$ $\overline{R_i}(l)<\rho$. As each node has access to $\overline{R_i}(l)$, convergence can be detected by each node and the convergence bit $b_i(lD+1)$ will be set to 1. Thus $b_i(lD+D+1)=1$ for all $i\in \V$ and algorithm will stop simultaneously at each node.\end{proof}

\begin{algorithm}[h]
    \SetKwBlock{Initialize}{Initialize:}{}
    \SetKwBlock{Input}{Input:}{}
    \SetKwBlock{Repeat}{Repeat:}{}
    \Input{$\rho$, $x^i(0)$ \tcp*{Initial condition}}
    \Initialize{$k := 0$; $R_i(0):= 0$; $y_i(0)=1$;$b_i(0)=0$; $l := 1$; }
    \Repeat {
    \Input{$x^j(k),y_j(k),R_j(k), j\in N^{-}_i$ }
    \tcc{ratio consensus updates of node $i$ given by (\ref{eq:numeratorState}), (\ref{eq:denominatorState}) and (\ref{eq:ratioState})}
        $x^{i}(k+1) := \sum\limits_{j\in \mathit{N^-_i}}p_{ji}(k)x^{j}(k)$; \\
        $y_{i}(k+1) := \sum\limits_{j\in \mathit{N^-_i}}p_{ji}(k)y_{j}(k)$;\\        
        $r^{i}(k+1) := \frac{1}{y_{i}(k+1)}x^{i}(k+1)$;\\
        \tcc{radius updates of node $i$ given by (\ref{eq:radiusUpdate}) }
        $R_i(k+1):=\underset{j\in  N^-_i}{\max}\{\|r^i(k+1)-r^j(k)\| + R_j(k)\}$\\
        \tcc{convergence bit update of node $i$ given by (\ref{eq:bitCons}) }
        $b_i(k+1)= \underset{j\in N^-_i}{\bigcup}b_j(k,)$\\
        \If {$ k=  lD$} {
        \uIf {$b_i(k+1)=1$} {\textbf{break} \tcp*{stop $x^i, y_i, r^i$, $R_i$ and $b_i$ updates}
            }
        \uElse {
            $\overline{R_i}(l)=R_i(k+1)$;\\
            \uIf {($\overline{R_i}(l) < \rho$)} {$b_i(k+1)=1$ \tcp*{set convergence bit to 1}
                }
            \uElse {$R_{i}(k+1) = 0$; $b_i(k+1)=0$;\\
             $l = l+1$;\\
            }
            } 
        }
        $k=k+1$;\\
    }\vspace{0.1in}
    \caption{Finite-time termination of ratio consensus in higher-dimension $d$ (at each node $i \in \V$)}
    \label{alg:radiusAlg}
\end{algorithm}

\begin{remark}\label{rem1}
Notice that using the above protocol, each node detects convergence simultaneously. Further, the only global parameter needed for Algorithm~\ref{alg:radiusAlg} is the knowledge of diameter $D$. However, it should be noted that an upper bound will suffice. In most applications, an upper bound on the diameter $D$ is readily available. 
\end{remark}
\begin{remark}\label{rem2}
It is to note here that for Algorithm~\ref{alg:radiusAlg}, only extra communication required between nodes is passing of the current radius at each node which is just a scalar along with a single bit for convergence consensus. Therefore the extra bandwidth required for each neighbor-neighbor interaction is $B+1$ where $B$ is the bit length (usually 32) for floating point representation. Thus, the above protocol is suitable for \textit{ad-hoc} communication networks where communication cost is high and bandwidth is limited.
\end{remark}

A finite-time termination criterion for vector consensus was previously provided in \cite{khatana2019gradient}. There, each element of ratio state required a maximum-minimum protocol (see (\ref{eq:maxvec}) and (\ref{eq:minvec})), with stopping criterion given by,
\begin{align*}
    \left\|\underset{s\in \{1,2,\dots,d\}}{\max}\underset{i\in \V}{\max}r^i_s(k)-\underset{s\in \{1,2,\dots,d\}}{\min}\underset{i\in \V}{\min}r^i_s(k) \right\|<\rho
\end{align*}
This maximum-minimum is a special case for finding a minimum convex set in the form of a hyper rectangle (box) which encompasses all the points. Here, at each iteration, two extra states are shared by each node, namely, one state for element-wise maximum and the other for element-wise minimum. Thus the extra communication bandwidth required for this algorithm is $2Bd$. An example case where $d=10, B=32$, requires an extra bandwidth of $640$ bits per interaction. For this example, Algorithm~\ref{alg:radiusAlg} only requires $B+1=33$ extra bits of communication per interaction, providing a reduction of more than $19$x. Thus for the applications with high dimensional vector consensus (like GANs, as described in the introduction, see \cite{goodfellow2014generative}), the algorithm reported here provides a reliable distributed stopping criterion with significantly less communication bandwidth.

\section{Applications of finite-time terminated average consensus in higher dimensions}

\subsection{Least Squares Estimation}

We follow \cite{garin2010survey} in our exposition of the least squares problem as solved by consensus.

Consider the problem of estimating a function $y = \varphi_\theta(x)$, given a noisy dataset $ \{(x_j, y_j) \}_{j=1}^N$ under the assumption that $\varphi_\theta(x)$ is a linear combination of known functions $g_i(x)$.  Explicitly $\varphi_\theta(x) = \sum_{i=1}^M \theta_i g_i(x)$.  Defining for $1 \leq j \leq N$, vectors $g^j = (g_1(x_j), \dots, g_M(x_j))^T$ and ${\bf G}$ to be the matrix formed by taking columns $G_j = g^j$. Then taking for granted the invertibility of the relevant matrices, the least squares estimate $\hat{\theta}$ of $\varphi_\theta$ is given by taking 
\begin{align} \label{eq: least squares estimator}
    \hat{\theta} &= \argmin_\theta | y - {\bf G} \theta |^2 
        \\
    &= ({\bf G^T} {\bf G})^{-1} {\bf G} y
        \\
    &= \left( \frac 1 N \sum_{j=1}^N g^j(g^j)^T \right)^{-1}\left( \frac 1 N \sum_{j=1}^N g^j y_j \right),
\end{align}
where the first equality is derived from setting the gradient of $\theta \mapsto |y - {\bf G} \theta|^2$ to zero and the second follows from algebra.

Initializing an average consensus algorithms on $N$ nodes, initialized to $f_j(0) = (M_j(0),z_j(0)) = (g^j(g^j)^T, g^j y_j)$, a node $i$, can form a consensus estimate of $\hat{\theta}$ at time $n$, by
\begin{align} \label{eq: consensus estimate of LSE}
    \theta_i(n) = {\bf M}_i^{-1}(n) z_i(n).
\end{align}
Indeed, as an average consensus algorithm, 
\begin{align*}
    {\bf M} \coloneqq \frac 1 N \sum_{j=1}^N g^j(g^j)^T = \lim_n {\bf M}_i(n), 
\end{align*}
and
\begin{align*}
    z \coloneqq \frac 1 N \sum_{j=1}^N g^j y_j = \lim_n z_i(n),
\end{align*}
thus
\begin{align*}
    \lim_n \theta_j(n) = \lim_n {\bf M}_i^{-1}(n) z_i(n) = {\bf M}^{-1} z = \hat{\theta}.
\end{align*}

For a matrix ${\bf A}$, let $\|{\bf A}\|_{op} = \inf \{ c: |{\bf A}x| \leq c|x| \}$

\begin{lem} \cite{horn2012matrix} \label{lem: Horn matrix lemma}
    For invertible matrices, ${\bf A}$ and ${\bf B}$,
    \begin{align*}
        \| {\bf A}^{-1} - {\bf B}^{-1} \|_{op} \leq \frac{\|{\bf A}^{-1}\|_{op}^2}{1 - \|{\bf A}^{-1}\|_{op} \| {\bf B}- {\bf A}\|_{op} } \|{\bf B}-{\bf A}\|_{op}
    \end{align*}
\end{lem}

\begin{theorem}
    The average consensus estimate $\theta_i(n)$ in~(\ref{eq: consensus estimate of LSE}), of the least squared estimator $\hat{\theta}$ in~(\ref{eq: least squares estimator}) satisfies
    \begin{align*}
        | \theta_i(n) - \hat{\theta} | \leq m |z_i(n) - z| + C \|{\bf M}_i(n) - {\bf M} \|_{op}
    \end{align*}
    where
    \begin{align*}
        m &= \| {\bf M}^{-1}_i(n)\|_{op}
            \\
        C &= \frac{m^2 ( |z_i(n)| + |z_i(n) - z|)}{ 1 - m \|{\bf M}_i(n) - {\bf M} \|_{op}}
    \end{align*}
\end{theorem}
Note the terms  $\|{\bf M}_i(n) - {\bf M}\|_{op}$ and $|z_i(n) - z|$ can be bounded through the finite time stopping criteria.  Further the terms $m$ and $C$ depend only on locally computable terms.  Thus the theorem demonstrates that not only can agents perform a distributed least square estimate, they can obtain error bounds on their estimates in a distributed fashion.  That $\lim_n m = \| {\bf M}^{-1} \|_{op}$ and $\lim_n C = \| {\bf M}^{-1} \|_{op}^2 |z|$ ensure convergence of the algorithm.

\begin{proof}
    For suppress we write ${\bf M}_i = {\bf M}_i(n)$ and $z_i = z_i(n)$ By definition, and the triangle inequality
    \begin{align*}
        | \theta_i(n) &- \hat{\theta} |
                \\
            &=
                | {\bf M}_i^{-1} z_i - {\bf M}^{-1} z|
                    \\
            &=
                |{\bf M}_i^{-1} (z_i-z) - ({\bf M}^{-1}-{\bf M}_i^{-1})(z)|
                    \\
            &\leq
                |{\bf M}_i^{-1} (z_i - z)| + |({\bf M}^{-1} - {\bf M}_i^{-1})(z)|.
    \end{align*}
    Applying the definition of the operator norm,
    \begin{align} 
        |{\bf M}_i^{-1} (z_i-z)|
            &\leq
                \| {\bf M}_i^{-1} \|_{op} |z_i-z|
                    \\
            &= 
                m |z_i -z|.\label{eq: first lse needed}
    \end{align}
        Further,
    \begin{align}
        |({\bf M}^{-1} &- {\bf M}_i^{-1})(z)|
            \\
            &\leq
                \|{\bf M}^{-1} - {\bf M}_i^{-1}\|_{op} |z|
                    \\
            &\leq
                \|{\bf M}^{-1} - {\bf M}_i^{-1}\|_{op} (|z_i| + |z- z_i|).  \label{eq: needs operator bound for inverses}
    \end{align}
    Applying Lemma \ref{lem: Horn matrix lemma} with ${\bf A} = {\bf M}_i$ and ${\bf B} = {\bf M}$,
    \begin{align} \label{eq: Horn applied}
        \|{\bf M}^{-1} - {\bf M}_i^{-1}\|_{op} 
            \leq
               \frac{\|{\bf M}_i^{-1}\|_{op}^2}{1 - \|{\bf M}_i^{-1}\|_{op} \| {\bf M}- {\bf M}_i\|_{op} } \|{\bf M}-{\bf M}_i\|_{op}.
    \end{align}
    Inserting \ref{eq: Horn applied} into \ref{eq: needs operator bound for inverses} gives
    \begin{align} \label{eq: second lse needed}
        |({\bf M}^{-1} &- {\bf M}_i^{-1})(z)| \leq C\|{\bf M}_i - {\bf M}\|_{op}.
    \end{align}
    Combining \eqref{eq: first lse needed} and \eqref{eq: second lse needed} gives our result.
\end{proof}

\subsection{Distributed Function Calculation}
Here, we give another application of the average consensus protocol in higher dimensions. We focus on the problem of computing arbitrary functions of agent state values over sensor networks, see \cite{giridhar2005computing}. In particular,  given a directed graph $\G(\V, \E)$ with $n$ nodes representing the communication constraints in a sensor network, the objective is to design an interaction rule for the nodes in the network to cooperatively compute a desired function $f(u_1(0), u_2(0), \dots, u_N(0))$, of the initial values $u_i(0) \in \mathbb{R}, i \in \{1, 2, \dots, N\}$. 
Such a problem is of interest in wireless sensor networks, see \cite{giridhar2005computing} where, the sink nodes in the sensor networks has to carry out the task of communicating
a relevant function of the raw sensor measurements. Another
example is the case of coordination tasks in multi-agent systems as given in \cite{ren2005survey},  \cite{olfati2007consensus}, where all agents communicate with each other to
coordinate their speed and direction of motion. Consider, the directed graph $\G(\V,\E)$ modeling the interconnection topology between the $n$ agents. Let each agent maintain three variables denoted by $x^i(k) \in \mathbb{R}^d, \ y_i(k) \in \mathbb{R}$ and $r^i(k) \in \mathbb{R}^d$ with the following initialization:
\begin{align}
    x^i(0) &= [0 \dots 0 \ Nu_i(0) \ 0 \dots 0]^T \in \mathbb{R}^d \label{eq:initialization1}\\
    y_i(0) &= 1 \label{eq:initialization2} \\
    r^i(0) &= \frac{1}{y_i(0)}{r^i(0)}. \label{eq:initialization3} 
\end{align}
The estimates $x^i(k)$, $y_i(k)$ and $r^i(k)$ are updated according to~(\ref{eq:numeratorState}), (\ref{eq:denominatorState}) and~(\ref{eq:ratioState}) respectively.
The following theorem upper bounds the error in distributed function calculation by the error in the consensus estimation.
\begin{theorem} \label{thm:func_cal}
Let $\G(\V,\E)$ and $P = [p_{ji}]$ associated with $\G(\V,\E)$ satisfy Assumptions~\ref{ass:StrConn} and~\ref{ass:colSto} respectively. Denote by $\{x^i(k)\}_{k\geq 0}, \{y_i(k)\}_{k\geq 0}$ and $\{r^i(k)\}_{k\geq 0}$ the sequences generated by~(\ref{eq:numeratorState}), (\ref{eq:denominatorState}) and~(\ref{eq:ratioState}) respectively. Under the initialization~(\ref{eq:initialization1})-(\ref{eq:initialization3}) the estimates $r^i(k)$ asymptotically converges to $\bar{r}:= \lim \limits_{k \rightarrow \infty} \frac{1}{y_i(k)}x^i(k) = [u_1(0) \ u_2(0) \dots u_N(0)]^T$ for all $i \in \{1,...,N\}$, and $f$ is $C$-Lipschitz, or more generally $\alpha$-H\"older continuous with constant $C$, then
\begin{align*}
    | f(r^i(k)) - f(\bar{r})| \leq   C \| r^i(k) - \bar{r} \|^\alpha.
\end{align*}col
\end{theorem}
\begin{proof} The proof follows from Theorem~\ref{thm:consensus_conv}. In particular, by Theorem~\ref{thm:consensus_conv} under Assumptions~\ref{ass:StrConn} and~\ref{ass:colSto} the updates~(\ref{eq:ratioState}) converges to $\overline{r}:= \lim \limits_{k \rightarrow \infty} \frac{1}{y_i(k)}x^i(k) = \frac{1}{N}\sum\limits_{j=1}^N x^j(0)$ for all $i \in \{1,...,N\}$. With the initialization~(\ref{eq:initialization1})-(\ref{eq:initialization3}) the limiting value $\overline{r}$ is given by
\begin{align*}
    \overline{r} = \frac{1}{N} \sum_{i=1}^N x^i(0) = [u_1(0) \ u_2(0) \dots u_N(0)]^T.
\end{align*}
To complete the proof one only needs to apply the definition of $\alpha$-H\"older continuity, that $|f(x) - f(y)| \leq C \| x- y\|^\alpha$, to the estimate $r^i(k)$ and the consensus value $\bar{r}$. \hspace{1.7in} \end{proof}

\begin{remark}
Theorem~\ref{thm:func_cal} guarantees that for large enough values of $k$ the agents following update rules~(\ref{eq:numeratorState})-(\ref{eq:ratioState}) will have enough information to calculate any arbitrary function of the initial values. Moreover, unlike the existing methods in the literature \cite{sundaram2008distributed} which require carefully designed matrices based on the global information of the network the proposed scheme allows for distributed synthesis. Further, the finite-time terminated protocol discussed here is applicable for arbitrary time-invariant connected directed graphs unlike the stringent assumptions required for applicability of some schemes in the literature, see \cite{kingston2006discrete}.
\end{remark}

\section{Results}\label{sec:results}

In this section, we present simulation results to demonstrate finite-time stopping criterion for high-dimensional ratio consensus. A network of 25 nodes is considered which is represented by a randomly generated directed graph (see Fig.~\ref{fig:graph25Nodes}(a)) with diameter $6$. Here the numerator state is chosen to be a 10-dimensional vector and selected randomly for every node. Equation (\ref{eq:numeratorState}), (\ref{eq:denominatorState}) and (\ref{eq:ratioState}) are implemented in MATLAB and simulated. 2-norm of each node's ratio state is plotted in Figure~\ref{fig:graph25Nodes}(b) achieving convergence in $60$ iterations. 

\begin{figure}[h]
    \centering
    \begin{tabular}{cc}
    \includegraphics[scale=0.5, trim={5cm 8cm 4.5cm 9cm}, clip]{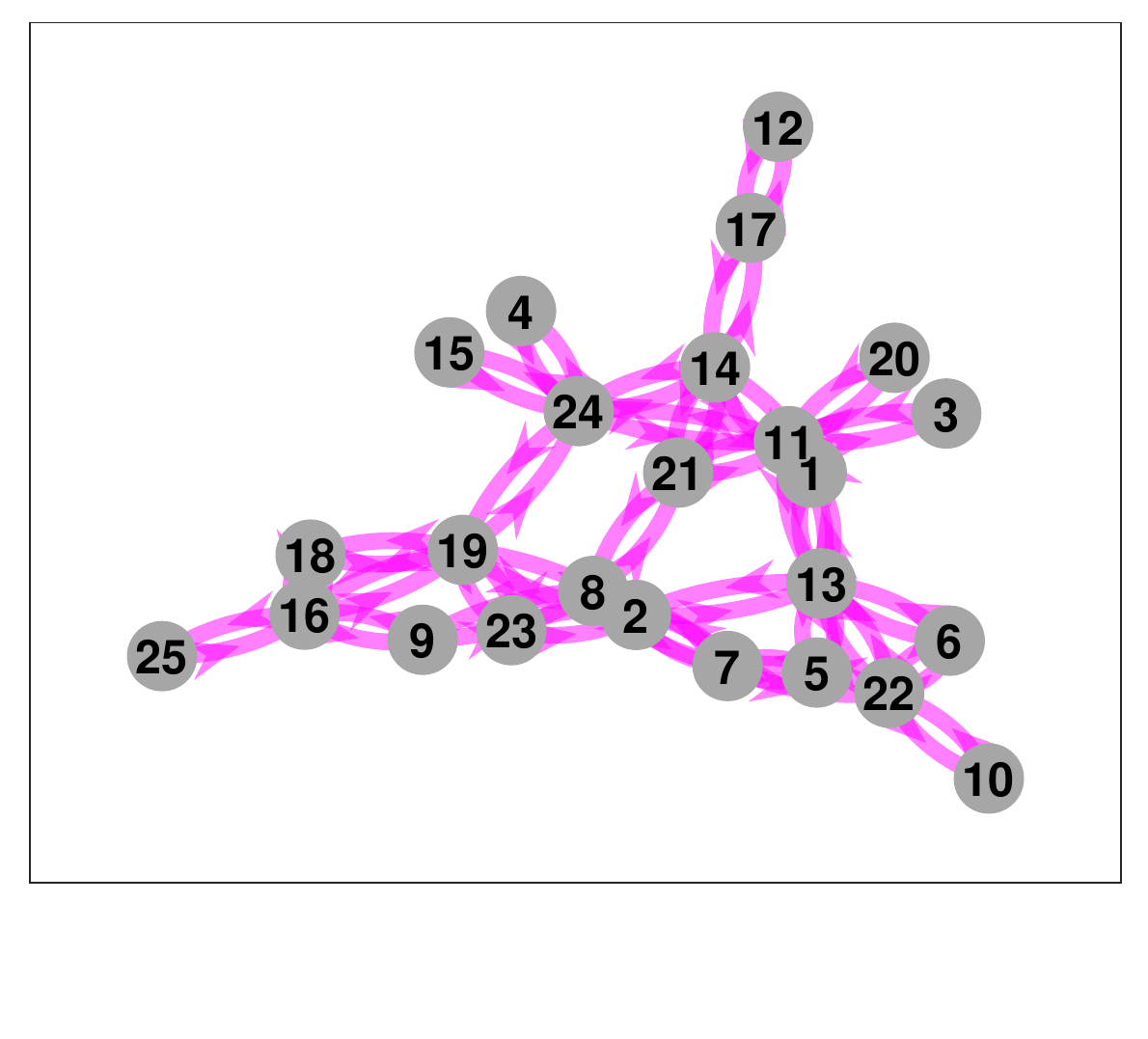} & \includegraphics[scale=0.5, trim={3.5cm 8cm 4.5cm 9cm}, clip]{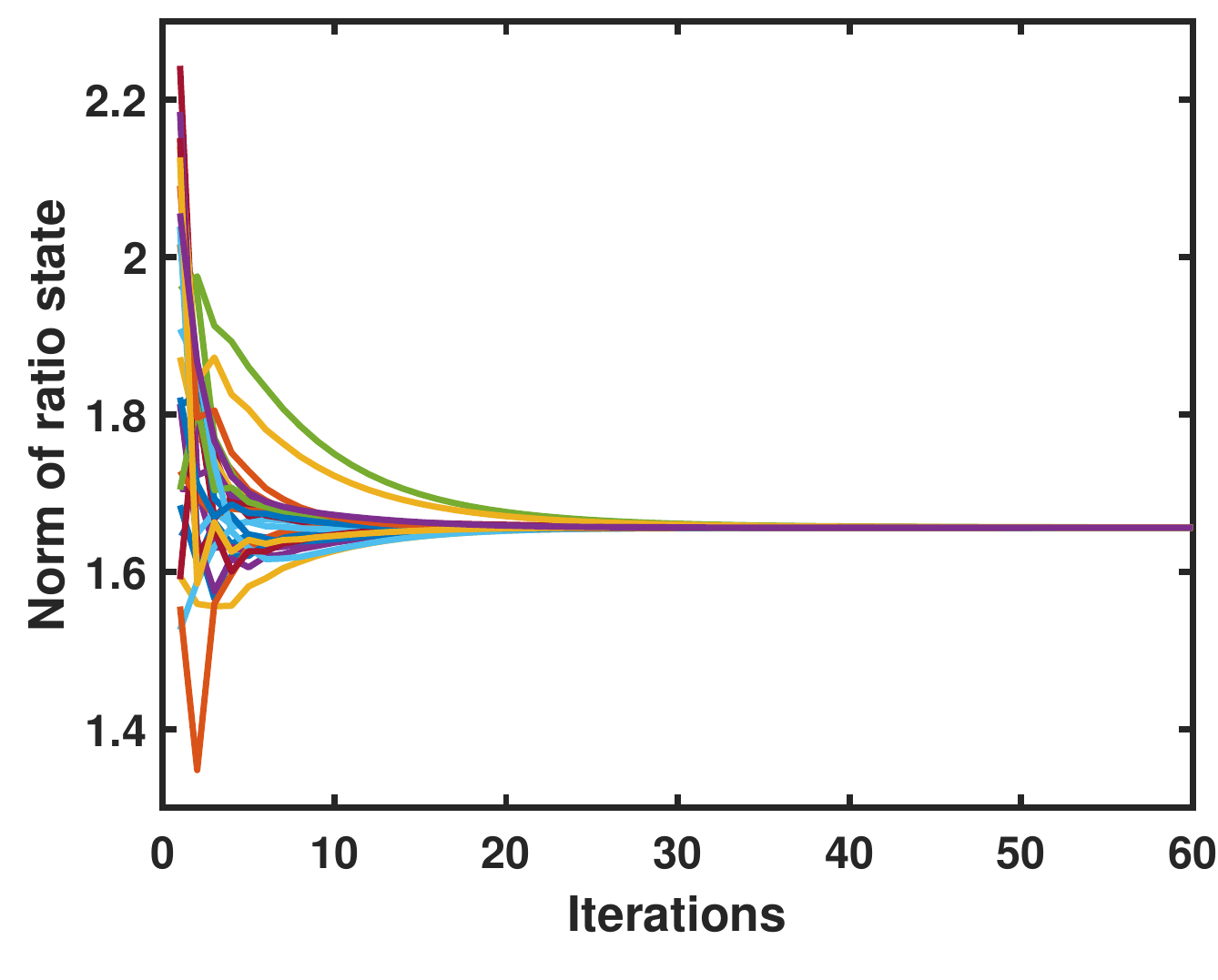}\\
        (a) & (b)
    \end{tabular}
    \caption{(a) A communication network represented by a $25$ node directed graph. (b) 2-norm of 10-dimensional ratio states of all the nodes (25) in the network.}
    \label{fig:graph25Nodes}
\end{figure}



Algorithm~\ref{alg:radiusAlg} is implemented in MATLAB and the radius $\overline{R_i}(l)$ for all $i \in \V$ is plotted in Fig.~\ref{fig:radiusNodes}(a). Here, it can be seen that radius comes under some pre-specified tolerance ($0.0166$, $1\%$ of the norm of the consensus vector) within 60 iteration and is used as a stopping criterion by each node. Fig.~\ref{fig:radiusNodes}(b) plots the two dimensional projection of the $norm$ ball $B\{\overline{R_i}(l),r^i(lD+D)\}$ for node $1$ as $l$ progresses over time. As expected, with increase in $l$, balls shrink in size. Similar observation is seen for all the other nodes as well.
The above illustration demonstrates how Algorithm~\ref{alg:radiusAlg} can be successfully used as finite-time termination criterion for distributed ratio consensus. 

\begin{figure}[h]
    \centering
    \begin{tabular}{cc}
    \includegraphics[scale=0.5, trim={3.5cm 8cm 4.5cm 9cm}, clip]{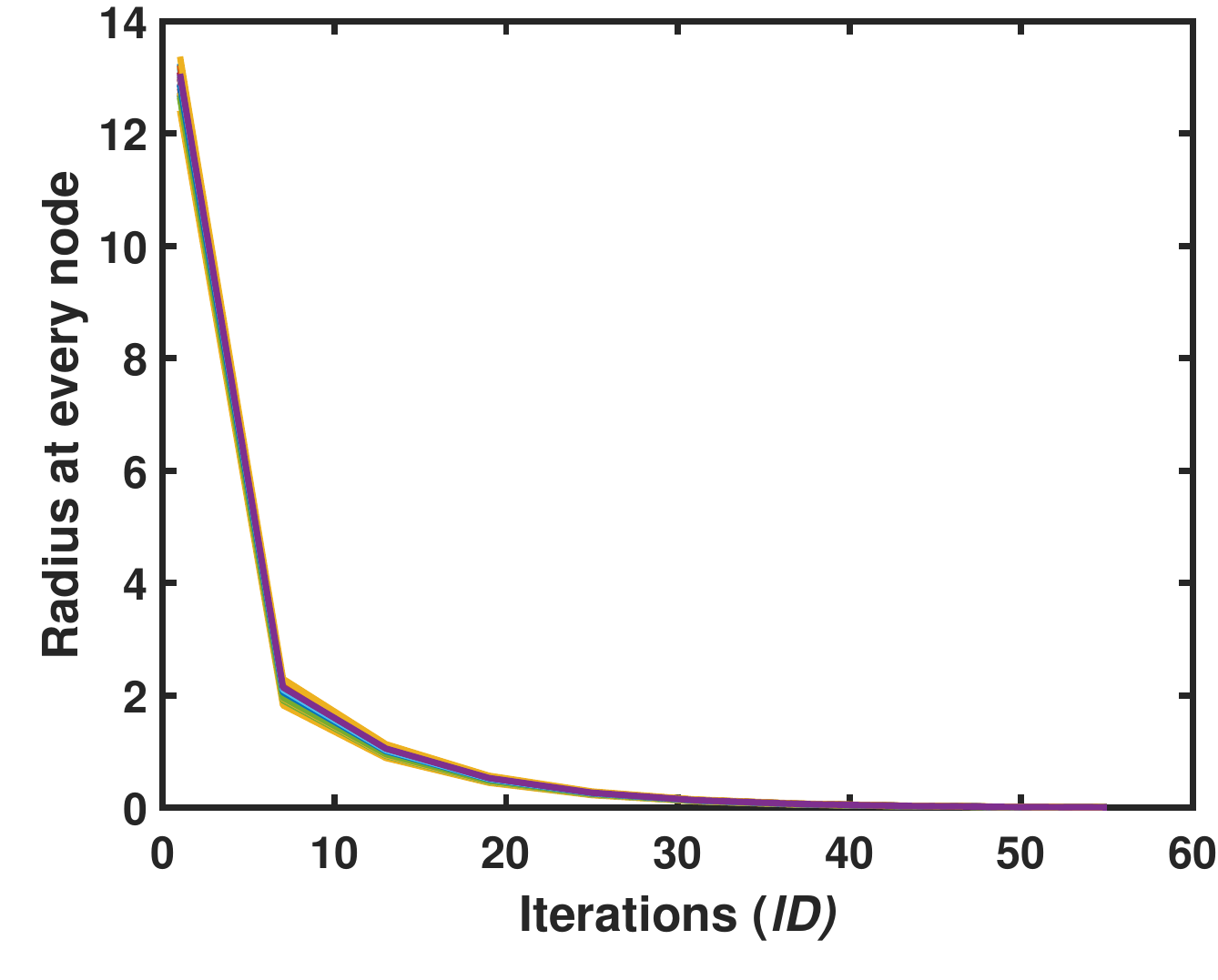} & 
    \includegraphics[scale=0.5, trim={3.5cm 8cm 4.5cm 9cm}, clip]{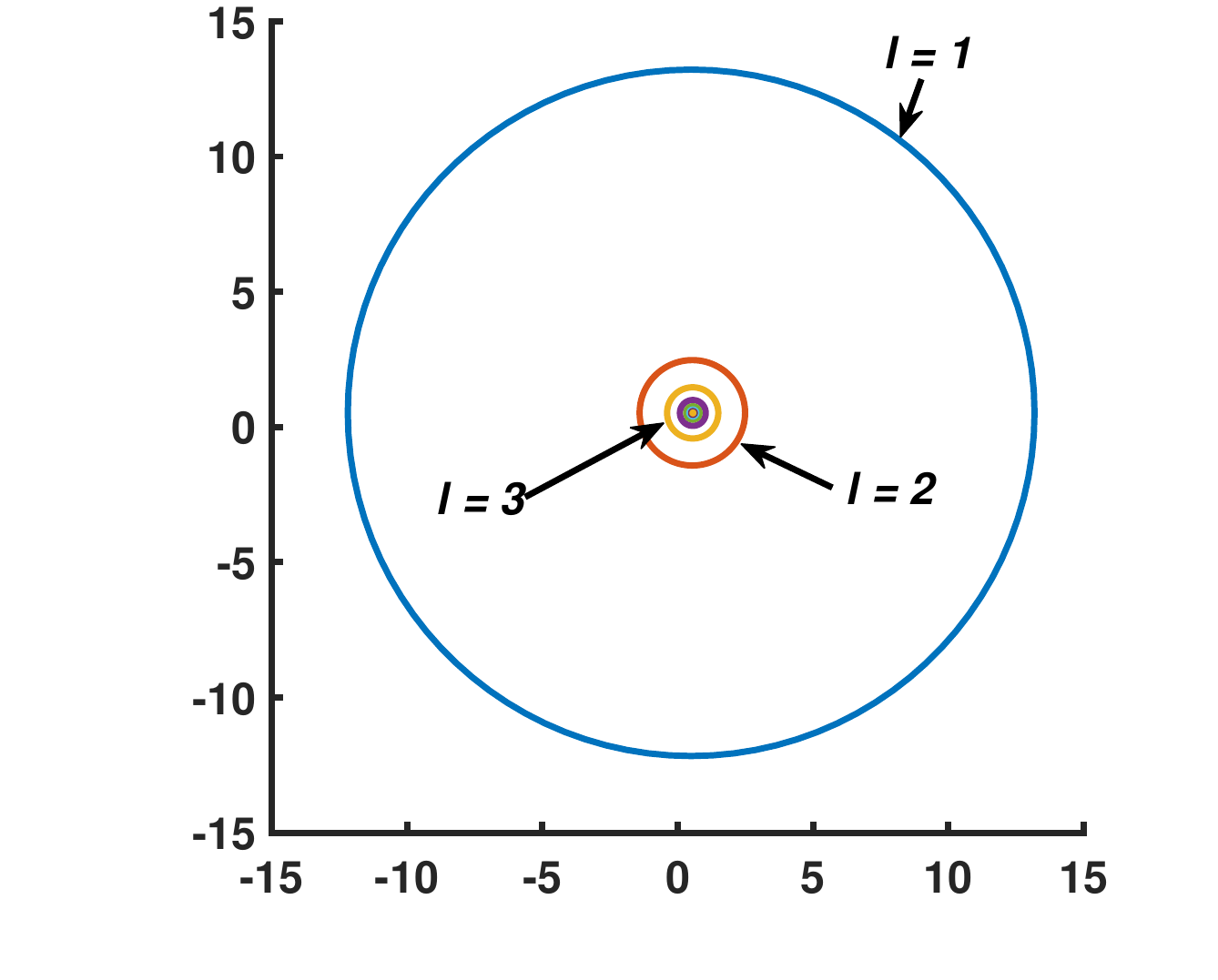}\\
        (a) & (b)
    \end{tabular}
    \caption{(a) Radius $\overline{R_i}(l)$ at each node. (b) 2-dimensional projection of norm balls for node $1$ with changing $l$.}
    \label{fig:radiusNodes}
\end{figure}

 

\section{Conclusion}\label{sec:Conclusion}
In this article, we presented a notion of monotonicity of network states in vector consensus algorithms, which we called a convex decreasing consensus algorithm. We showed that this property can be used to construct finite-time stopping criterion and provided a distributed algorithm. We further provided an algorithm which calculates an approximation of minimum norm balls which contain all the network states at a given iteration. Radius of these balls was shown to converge to zero, and algorithm was presented to use that as a finite-time stopping criterion. This algorithm was shown to have much smaller communication requirement compared to existing methods.
The effectiveness of our algorithm is validated by simulating a vector ($\in \mathbb{R}^{10}$) ratio consensus algorithm  for a network graph of 25 nodes.  Further we demonstrated how these stopping criteria could be applied to provide guarantees on the convergence of Least Squared Estimator approximation through consensus.

\bibliographystyle{ieeetr} 

\bibliography{topident}

\begin{thebibliography}{10}

\bibitem{arrow1958decentralization}
K.~J. Arrow and L.~Hurwicz, {\em Decentralization and computation in resource
  allocation}.
\newblock Stanford University, Department of Economics, 1958.

\bibitem{degroot1974reaching}
M.~H. DeGroot, ``Reaching a consensus,'' {\em Journal of the American
  Statistical Association}, vol.~69, no.~345, pp.~118--121, 1974.

\bibitem{lynch1996distributed}
N.~A. Lynch, {\em Distributed algorithms}.
\newblock Elsevier, 1996.

\bibitem{tsitsiklis1984problems}
J.~N. Tsitsiklis, ``Problems in decentralized decision making and
  computation.,'' tech. rep., DTIC Document, 1984.

\bibitem{kempe2003gossip}
D.~Kempe, A.~Dobra, and J.~Gehrke, ``Gossip-based computation of aggregate
  information,'' in {\em 44th Annual IEEE Symposium on Foundations of Computer
  Science, 2003. Proceedings.}, pp.~482--491, IEEE, 2003.

\bibitem{dominguez2010coordination}
A.~D. Dominguez-Garcia and C.~N. Hadjicostis, ``Coordination and control of
  distributed energy resources for provision of ancillary services,'' in {\em
  Smart Grid Communications (SmartGridComm), 2010 First IEEE International
  Conference on}, pp.~537--542, IEEE, 2010.

\bibitem{hadjicostis2013average}
C.~N. Hadjicostis and T.~Charalambous, ``Average consensus in the presence of
  delays in directed graph topologies,'' {\em IEEE Transactions on Automatic
  Control}, vol.~59, no.~3, pp.~763--768, 2013.

\bibitem{cai2012average}
K.~Cai and H.~Ishii, ``Average consensus on general strongly connected
  digraphs,'' {\em Automatica}, vol.~48, no.~11, pp.~2750--2761, 2012.

\bibitem{predd2009collaborative}
J.~B. Predd, S.~R. Kulkarni, and H.~V. Poor, ``A collaborative training
  algorithm for distributed learning,'' {\em IEEE Transactions on Information
  Theory}, vol.~55, no.~4, pp.~1856--1871, 2009.

\bibitem{fax2002information}
J.~A. Fax and R.~M. Murray, ``Information flow and cooperative control of
  vehicle formations,'' {\em IFAC Proceedings Volumes}, vol.~35, no.~1,
  pp.~115--120, 2002.

\bibitem{olfati2007consensus}
R.~Olfati-Saber, J.~A. Fax, and R.~M. Murray, ``Consensus and cooperation in
  networked multi-agent systems,'' {\em Proceedings of the IEEE}, vol.~95,
  no.~1, pp.~215--233, 2007.

\bibitem{nedic2014distributed}
A.~Nedi{\'c} and A.~Olshevsky, ``Distributed optimization over time-varying
  directed graphs,'' {\em IEEE Transactions on Automatic Control}, vol.~60,
  no.~3, pp.~601--615, 2014.

\bibitem{khatana2019gradient}
V.~Khatana, G.~Saraswat, S.~Patel, and M.~V. Salapaka, ``Gradient-consensus
  method for distributed optimization in directed multi-agent networks,'' {\em
  arXiv preprint arXiv:1909.10070}, 2019.

\bibitem{khan2009distributed}
U.~A. Khan, S.~Kar, and J.~M. Moura, ``Distributed sensor localization in
  random environments using minimal number of anchor nodes,'' {\em IEEE
  Transactions on Signal Processing}, vol.~57, no.~5, pp.~2000--2016, 2009.

\bibitem{khan2010higher}
U.~A. Khan, S.~Kar, and J.~M. Moura, ``Higher dimensional consensus: Learning
  in large-scale networks,'' {\em IEEE Transactions on Signal Processing},
  vol.~58, no.~5, pp.~2836--2849, 2010.

\bibitem{patel2017distributed}
S.~Patel, S.~Attree, S.~Talukdar, M.~Prakash, and M.~V. Salapaka, ``Distributed
  apportioning in a power network for providing demand response services,'' in
  {\em 2017 IEEE International Conference on Smart Grid Communications
  (SmartGridComm)}, pp.~38--44, IEEE, 2017.

\bibitem{nedic2010constrained}
A.~Nedic, A.~Ozdaglar, and P.~A. Parrilo, ``Constrained consensus and
  optimization in multi-agent networks,'' {\em IEEE Transactions on Automatic
  Control}, vol.~55, no.~4, pp.~922--938, 2010.

\bibitem{chakraborty2017deep}
S.~Chakraborty, A.~Preece, M.~Alzantot, T.~Xing, D.~Braines, and M.~Srivastava,
  ``Deep learning for situational understanding,'' in {\em 2017 20th
  International Conference on Information Fusion (Fusion)}, pp.~1--8, IEEE,
  2017.

\bibitem{goodfellow2014generative}
I.~Goodfellow, J.~Pouget-Abadie, M.~Mirza, B.~Xu, D.~Warde-Farley, S.~Ozair,
  A.~Courville, and Y.~Bengio, ``Generative adversarial nets,'' in {\em
  Advances in neural information processing systems}, pp.~2672--2680, 2014.

\bibitem{kurakin2016adversarial}
A.~Kurakin, I.~Goodfellow, and S.~Bengio, ``Adversarial machine learning at
  scale,'' {\em ICLR}, 2017.

\bibitem{li2009distributed}
Z.~Li, F.~R. Yu, and M.~Huang, ``A distributed consensus-based cooperative
  spectrum-sensing scheme in cognitive radios,'' {\em IEEE Transactions on
  Vehicular Technology}, vol.~59, no.~1, pp.~383--393, 2009.

\bibitem{patel2020distributed}
S.~Patel, V.~Khatana, G.~Saraswat, and M.~V. Salapaka, ``Distributed detection
  of malicious attacks on consensus algorithms with applications in power
  networks,'' 2020.

\bibitem{yadav2007distributed}
V.~Yadav and M.~V. Salapaka, ``Distributed protocol for determining when
  averaging consensus is reached,'' in {\em 45th Annual Allerton Conf},
  pp.~715--720, 2007.

\bibitem{ratio_consensus_lab}
M.~Prakash, S.~Talukdar, S.~Attree, S.~Patel, and M.~V. Salapaka, ``Distributed
  {{Stopping Criterion}} for {{Ratio Consensus}},'' in {\em 2018 56th {{Annual
  Allerton Conference}} on {{Communication}}, {{Control}}, and {{Computing}}
  ({{Allerton}})}, pp.~131--135, Oct. 2018.

\bibitem{saraswat2019distributed}
G.~Saraswat, V.~Khatana, S.~Patel, and M.~V. Salapaka, ``Distributed
  finite-time termination for consensus algorithm in switching topologies,''
  {\em arXiv preprint arXiv:1909.00059}, 2019.

\bibitem{prakash2019distributed}
M.~Prakash, S.~Talukdar, S.~Attree, V.~Yadav, and M.~V. Salapaka, ``Distributed
  stopping criterion for consensus in the presence of delays,'' {\em IEEE
  Transactions on Control of Network Systems}, 2019.

\bibitem{sundaram2007finite}
S.~Sundaram and C.~N. Hadjicostis, ``Finite-time distributed consensus in
  graphs with time-invariant topologies,'' in {\em 2007 American Control
  Conference}, pp.~711--716, IEEE, 2007.

\bibitem{preparata1979optimal}
F.~P. Preparata, ``An optimal real-time algorithm for planar convex hulls,''
  {\em Communications of the ACM}, vol.~22, no.~7, pp.~402--405, 1979.

\bibitem{fernandez2017one}
D.~Fern{\'a}ndez-Francos, {\'O}.~Fontenla-Romero, and A.~Alonso-Betanzos,
  ``One-class convex hull-based algorithm for classification in distributed
  environments,'' {\em IEEE Transactions on Systems, Man, and Cybernetics:
  Systems}, 2017.

\bibitem{casale2014approximate}
P.~Casale, O.~Pujol, and P.~Radeva, ``Approximate polytope ensemble for
  one-class classification,'' {\em Pattern Recognition}, vol.~47, no.~2,
  pp.~854--864, 2014.

\bibitem{kavan2006fast}
L.~Kavan, I.~Kolingerova, and J.~Zara, ``Fast approximation of convex hull.,''
  {\em ACST}, vol.~6, pp.~101--104, 2006.

\bibitem{osuna2002convex}
E.~Osuna and O.~De~Castro, ``Convex hull in feature space for support vector
  machines,'' in {\em Ibero-American Conference on Artificial Intelligence},
  pp.~411--419, Springer, 2002.

\bibitem{kim2015distributed}
W.~Kim, M.~S. Stankovi{\'c}, K.~H. Johansson, and H.~J. Kim, ``A distributed
  support vector machine learning over wireless sensor networks,'' {\em IEEE
  transactions on cybernetics}, vol.~45, no.~11, pp.~2599--2611, 2015.

\bibitem{giridhar2005computing}
A.~Giridhar and P.~R. Kumar, ``Computing and communicating functions over
  sensor networks,'' {\em IEEE Journal on selected areas in communications},
  vol.~23, no.~4, pp.~755--764, 2005.

\bibitem{sundaram2008distributed}
S.~Sundaram and C.~N. Hadjicostis, ``Distributed function calculation and
  consensus using linear iterative strategies,'' {\em IEEE journal on selected
  areas in communications}, vol.~26, no.~4, pp.~650--660, 2008.

\bibitem{melbourne2020geometry}
J.~Melbourne, G.~Saraswat, V.~Khatana, S.~Patel, and M.~V. Salapaka, ``On the
  geometry of consensus algorithms with application to distributed termination
  in higher dimension,'' {\em the proceedings of International Federation of
  Automatic Control (IFAC)}, 2020.

\bibitem{Die06}
R.~Diestel, {\em Graph Theory}.
\newblock Berlin, Germany: Springer-Verlag, 2006.

\bibitem{horn2012matrix}
R.~A. Horn and C.~R. Johnson, {\em Matrix analysis}.
\newblock Cambridge university press, 2012.

\bibitem{rockafellar1970convex}
R.~T. Rockafellar, {\em Convex analysis}.
\newblock No.~28, Princeton university press, 1970.

\bibitem{Lax2002functionalanalysis}
P.~Lax, {\em Functional Analysis}, vol.~1.
\newblock Wiley-Interscience, 2002.

\bibitem{levin2017markov}
D.~A. Levin and Y.~Peres, {\em Markov chains and mixing times}, vol.~107.
\newblock American Mathematical Soc., 2017.

\bibitem{gray2009probability}
R.~M. Gray and R.~Gray, {\em Probability, random processes, and ergodic
  properties}, vol.~1.
\newblock Springer, 2009.

\bibitem{clarkson1989applications}
K.~L. Clarkson and P.~W. Shor, ``Applications of random sampling in
  computational geometry, ii,'' {\em Discrete \& Computational Geometry},
  vol.~4, no.~5, pp.~387--421, 1989.

\bibitem{barber1996quickhull}
C.~B. Barber, D.~P. Dobkin, D.~P. Dobkin, and H.~Huhdanpaa, ``The quickhull
  algorithm for convex hulls,'' {\em ACM Transactions on Mathematical Software
  (TOMS)}, vol.~22, no.~4, pp.~469--483, 1996.

\bibitem{fischer1975smallest}
K.~Fischer, {\em Smallest enclosing balls of balls}.
\newblock PhD thesis, ETH Z{\"u}rich, 1975.

\bibitem{garin2010survey}
F.~Garin and L.~Schenato, ``A survey on distributed estimation and control
  applications using linear consensus algorithms,'' in {\em Networked control
  systems}, pp.~75--107, Springer, 2010.

\bibitem{ren2005survey}
W.~Ren, R.~W. Beard, and E.~M. Atkins, ``A survey of consensus problems in
  multi-agent coordination,'' in {\em Proceedings of the 2005, American Control
  Conference, 2005.}, pp.~1859--1864, IEEE, 2005.

\bibitem{kingston2006discrete}
D.~B. Kingston and R.~W. Beard, ``Discrete-time average-consensus under
  switching network topologies,'' in {\em 2006 American Control Conference},
  pp.~6--pp, IEEE, 2006.

\end{thebibliography}

\begin{IEEEbiography}[
{
\includegraphics[width=1in,height=1.25in,clip,keepaspectratio]{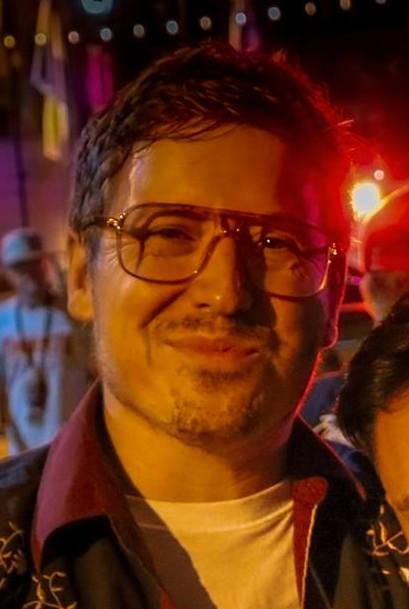}
}
]
{James Melbourne} received his Bachelors in Art History in 2006 and a Masters in Mathematics in 2009 both from the University of Kansas, and his PhD in Mathematics in 2015 at the University of Minnesota.   He was a postdoctoral researcher in the University of Delaware Mathematics department from 2015 to 2017, and is currently a postdoctoral researcher at the University of Minnesota in Electrical and Computer Engineering.  His research interest include convexity theory, particularly its application to probabilistic, geometric, and information theoretic inequalities, consensus algorithms, and stochastic energetics.

\end{IEEEbiography}

\begin{IEEEbiography}[
{
\includegraphics[width=1in,height=1.25in,clip,keepaspectratio]{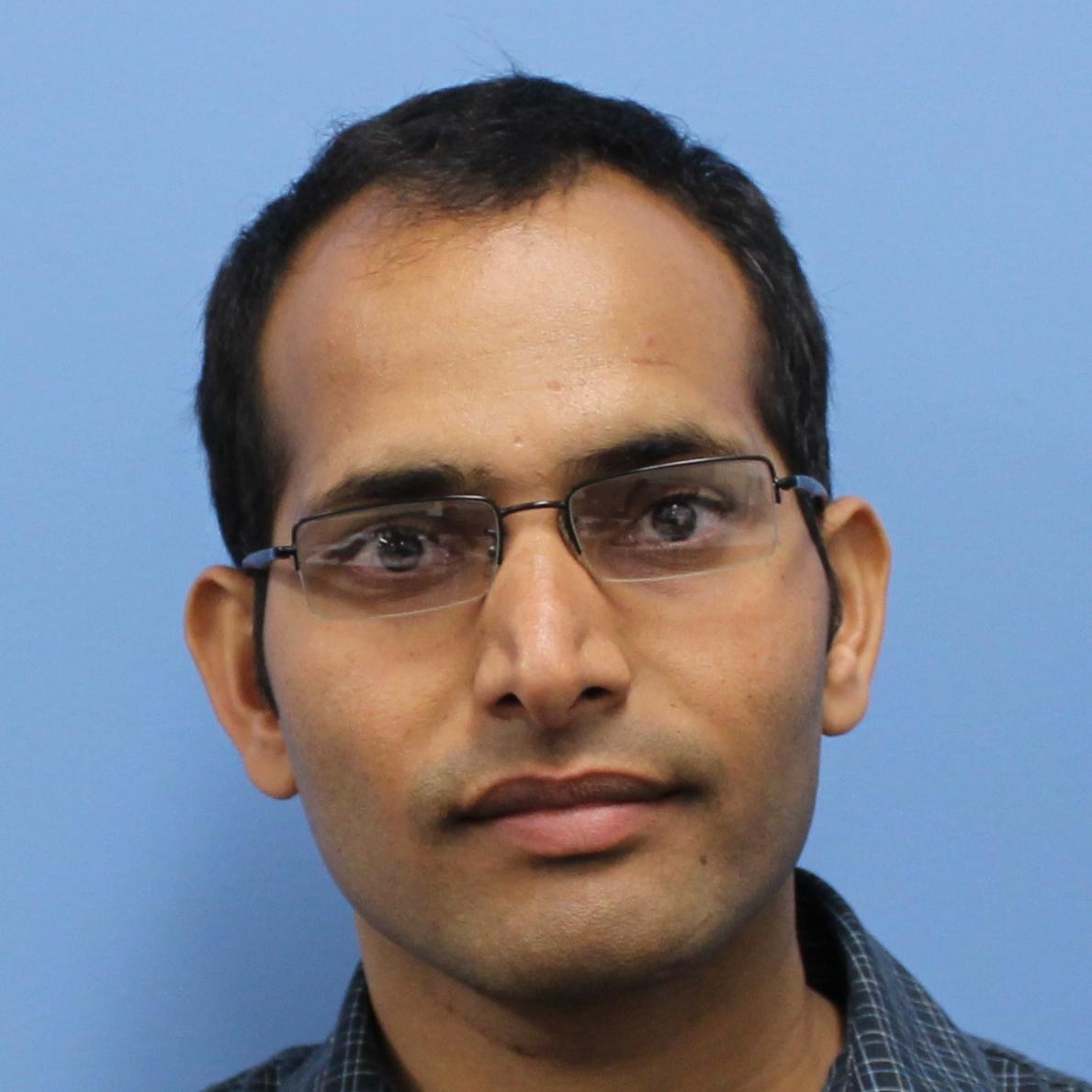}
}
]
{Govind Saraswat} received his B.Tech degree in Electrical Engineering from the Indian Institute of Technology, Delhi, in 2007 and his PhD degree in Electrical Engineering from University of Minnesota, Twin Cities in 2014. Currently, he is part of Sensing and Predictive Analytics group at National Renewable Energy Laboratory, Golden, CO (NREL) where he works on data-driven technology for energy systems planning and operation. His research includes power system modeling and analysis, measurement-based operation and control, machine learning, and optimization.
\end{IEEEbiography}

\begin{IEEEbiography}[
{
\includegraphics[width=1in,height=1.25in,clip,keepaspectratio]{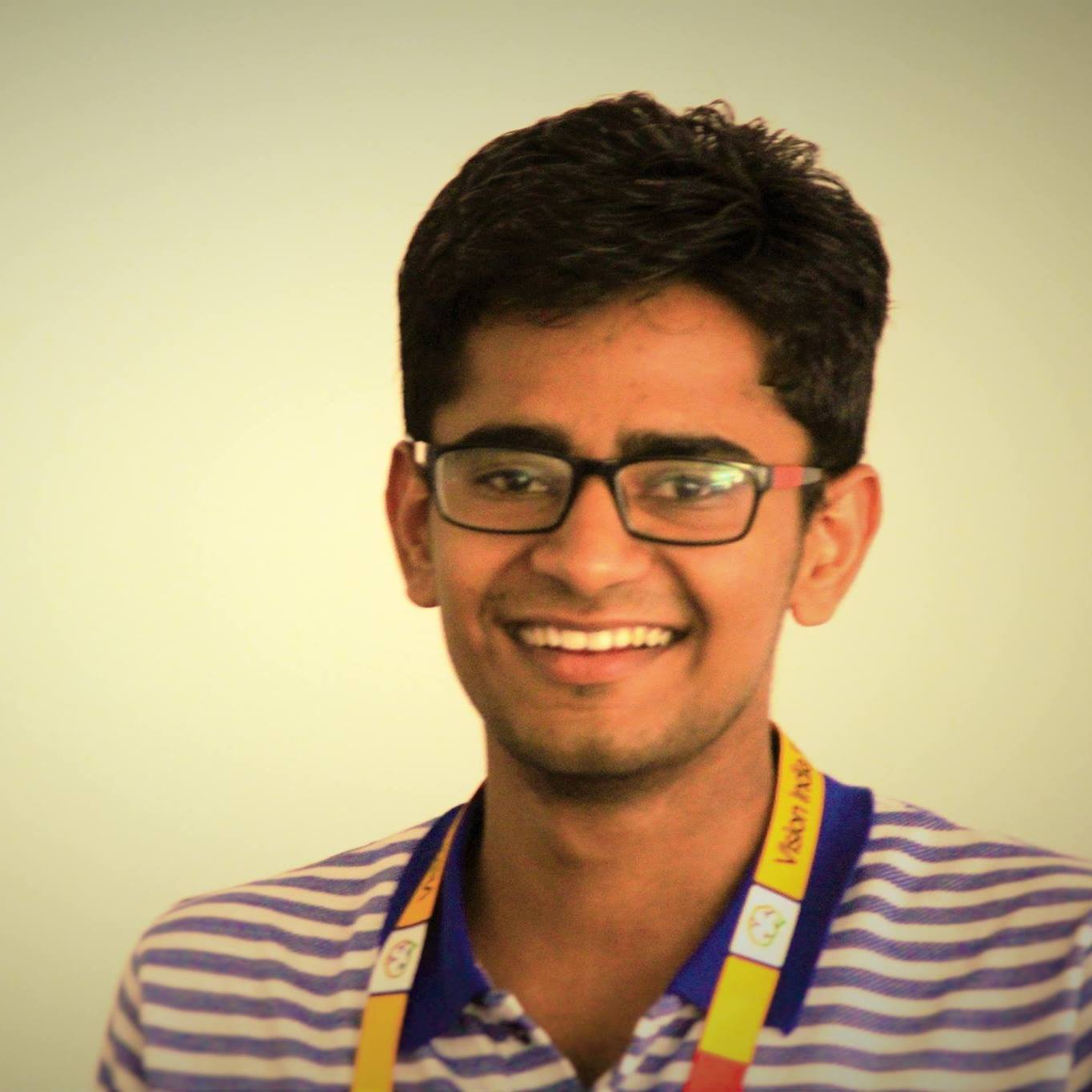}
}
]
{Vivek Khatana} received the B.Tech degree in Electrical Engineering from the Indian Institute of Technology, Roorkee, in 2018. Currently, he is working towards a Ph.D. degree at the department of Electrical Engineering at University of Minnesota. His Ph.D. research interest includes distributed optimization, consensus algorithms, distributed control and stochastic calculus. 
\end{IEEEbiography}

\begin{IEEEbiography}[
{
\includegraphics[width=1in,height=1.25in,clip,keepaspectratio]{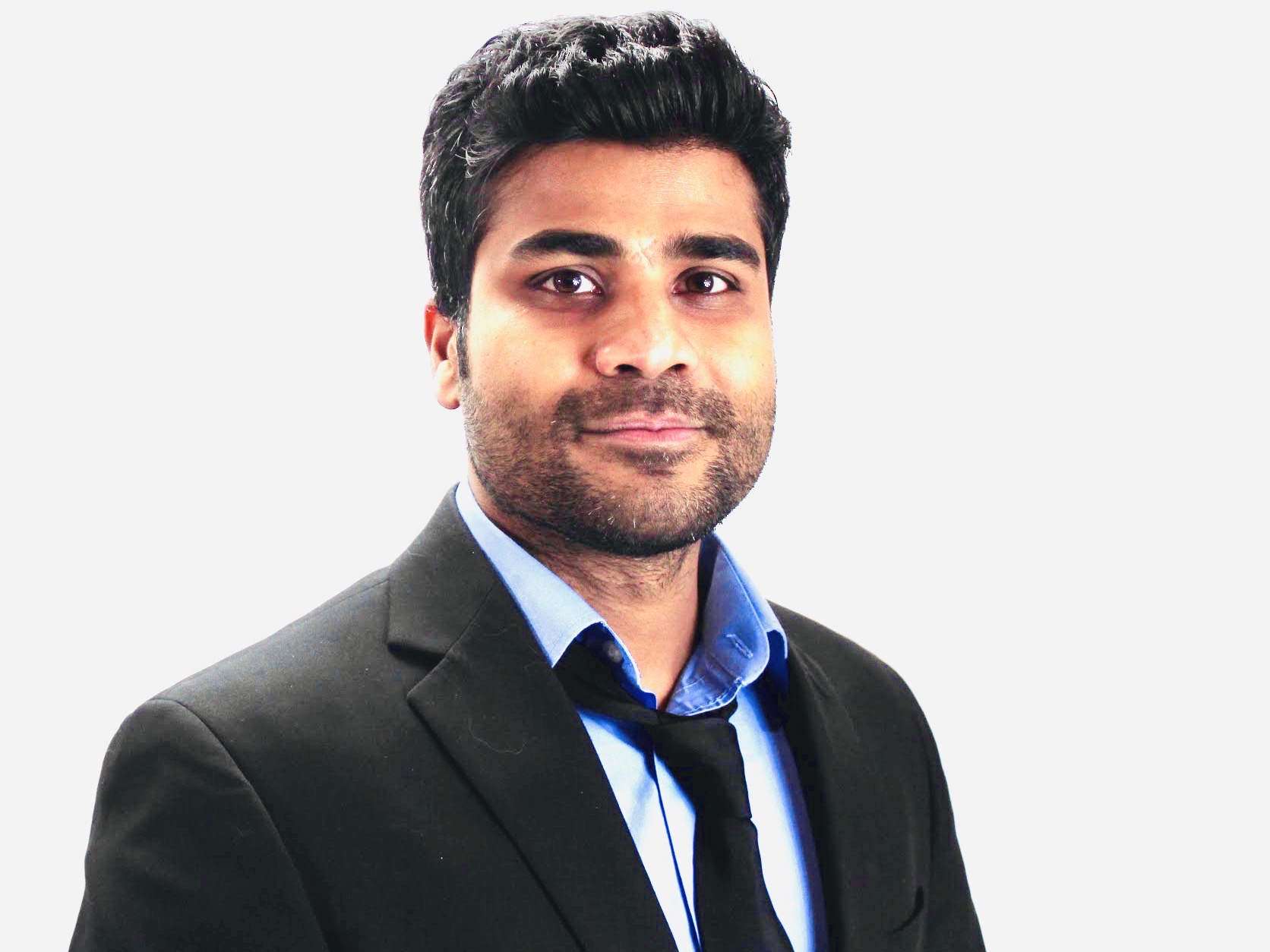}
}
]
{Sourav Kumar Patel}
received his B.Tech. degree in Instrumentation and Control Engineering from the National Institute of Technology, Jalandhar, in 2011. In the same year, he joined National Thermal Power Corporation Ltd. in Kaniha, Odisha, India as a Control and Instrumentation Engineer. He received his M.S. degree in Electrical Engineering from the University of Minnesota, Twin-Cities, in 2018 where currently, he is working towards his Ph.D. degree in Electrical Engineering. His research interests include control and systems theory, coordination and communication protocols for Distributed Energy Resources towards smart grid applications.
\end{IEEEbiography}

\begin{IEEEbiography}[
{
\includegraphics[width=1in,height=1.25in,clip,keepaspectratio]{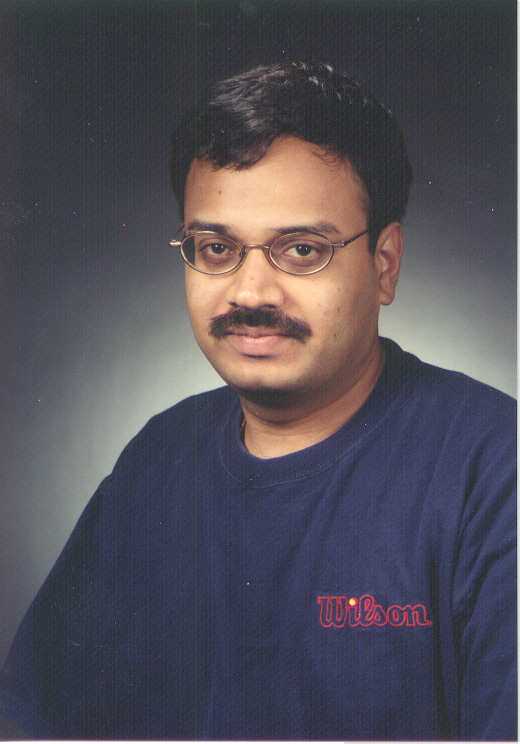}
}
]
{Murti V. Salapaka}

received the B.Tech. degree in Mechanical Engineering from the Indian Institute of Technology, Madras, in 1991 and the M.S. and Ph.D. degrees in Mechanical Engineering from the University of California at Santa Barbara, in 1993 and 1997, respectively. He was a faculty member in the Electrical and Computer Engineering Department, Iowa State University, Ames, from 1997 to 2007. Currently, he is the Director of Graduate Studies and the Vincentine Hermes Luh Chair Professor in the Electrical and Computer Engineering Department, University of Minnesota, Minneapolis. His research interests include control and network science, nanoscience and single molecule physics. Dr. Salapaka received the 1997 National Science Foundation CAREER Award and is an IEEE fellow.

\end{IEEEbiography}

\end{document}